\documentclass[letterpaper,10pt,twoside,journal,tbtags]{IEEEtran}
\usepackage[T1]{fontenc}
\usepackage{xcolor}
\usepackage{mathtools, amssymb, amsfonts, dsfont}
\usepackage{microtype}
\usepackage{url}
\usepackage{hyperref}
\usepackage[sort,compress]{cite}

\usepackage[standard,amsmath,thmmarks]{ntheorem}
\usepackage{xcolor}
\usepackage{booktabs}
\usepackage{subcaption}
\newcommand*\dash{-\penalty0\hskip0pt\relax}%

\usepackage{tikz}
\tikzset{
  agent/.style={draw, fill=blue!25!white,circle, minimum size=2mm,
  font=\footnotesize},
  x = 1.5cm, y=1.5cm,
  every loop/.style={},
}
\usetikzlibrary{matrix}

\newtheorem{problem}{Problem}

\newcommand*\BR{\mathds{R}}
\newcommand*\BN{\mathds{N}}

\newcommand*\CN{\mathcal{N}}

\newcommand*\TRANS{{\mathpalette\doTRANS\empty}}
\makeatletter
\newcommand*\doTRANS[2]{\raisebox{\depth}{$\m@th#1\intercal$}}
\makeatother

\newcommand\DEFINED{\coloneqq}

\DeclareMathOperator\COLS{cols}
\DeclareMathOperator\DIAG{diag}

\DeclareMathOperator\Tr{Tr}

\newcommand\MATRIX[1]{\begin{bmatrix} #1 \end{bmatrix}}
\newcommand\SMATRIX[1]{\left[\begin{smallmatrix} #1 \end{smallmatrix}\right]}

\begin{document}

\title{{{}Optimal} control of {{}network}-coupled  subsystems: Spectral
decomposition and low-dimensional solutions}
\author{Shuang Gao~\IEEEmembership{Member,~IEEE,} and Aditya Mahajan~\IEEEmembership{Senior Member,~IEEE}%
\thanks{S. Gao and A. Mahajan are with the Department of Electrical and Computer Engineering, McGill University,
  Montreal, QC, Canada. 
  Emails: {\tt\small  sgao@cim.mcgill.ca, aditya.mahajan@mcgill.ca}.}%
  \thanks{The work of the second author was supported in part by  the
    Innovation for Defence Excellence and Security (IDEaS) Program of the
  Canadian Department of National Defence through grant CFPMN2-30}%
  \thanks{Preliminary version of this work was presented at the 58th IEEE Conference on Decision and Control, Nice, France, December, 2019.}
}
\maketitle
\begin{abstract}
  In this paper, we investigate optimal control of {{}network}-coupled
 subsystems where the dynamics and the cost couplings depend on an underlying undirected 
  weighted graph. 
 {The graph coupling matrix {{}in the dynamics} {may} be the adjacency matrix, the Laplacian matrix, or any other symmetric matrix corresponding to the underlying graph.} {{}The cost couplings can be any polynomial function of the underlying coupling matrix.}
  We use the spectral decomposition of the graph  {coupling}
  matrix to decompose the overall system into $({L+1})$ systems with decoupled
  dynamics and cost, where $L$ is the rank of the coupling matrix.  
{{}Furthermore,} the optimal control input at each subsystem can be computed by
  solving $({L_\textup{dist}+1})$ decoupled Riccati equations {}{where $L_\textup{dist}~(L_\textup{dist}\leq L)$ is the number of distinct non-zero eigenvalues of the coupling matrix}. A salient feature of the
  result is that the solution complexity depends {}{on the number of distinct eigenvalues of the coupling
  matrix} rather than the size of the network.
  Therefore, the proposed solution framework provides a scalable method for
  synthesizing and implementing optimal control laws for large-scale {{}network-coupled subsystems}.
\end{abstract}

\section{Introduction}
\subsection{Motivation}
The recent proliferation of low cost sensors and actuators has given rise to 
many networked control systems such as the Internet of Things, smart grids, smart buildings,
etc.,\@ where multiple subsystems are connected over a network. In such
systems, the evolution of the state of a subsystem depends on its local state
and local control input and is also influenced by the states and controls of its
neighbors. 
Such networks are often referred to as large-scale systems
or complex networks, and various aspects of such systems have been investigated since the early
1970s~\cite{ho1976directions,sandell1978survey}, including issues such as
controllability~\cite{liu2011controllability,yang1995structural}, observability~\cite{yan2015spectrum,yang1995structural}, control energy metric~\cite{pasqualetti2014controllability}, distributed control~\cite{saber2003consensus,olfati2007consensus,jadbabaie2003coordination,movric2013cooperative,lin2004local}, decentralized
control~\cite{sundareshan1991qualitative,yang1996decentralized,vsiljak2005control,bian2015decentralized}
and adaptive control~\cite{bian2015decentralized}.

A key theme for investigating large-scale  {networked control} systems is to identify conditions
under which the optimal control laws may be synthesized and implemented with
low-complexity. Such conditions include simplified control objectives (e.g.,
  consensus~\cite{saber2003consensus,olfati2007consensus,jadbabaie2003coordination,movric2013cooperative,lin2004local} and
synchronization~\cite{arenas2008synchronization}), simplified control inputs
(e.g., pinning control~\cite{gang1994controlling, grigoriev1997pinning,
  wang2002pinning} and ensemble
control~\cite{li2011ensemble}), simplified coupling between subsystems (e.g.,
  symmetric interconnections~\cite{lunze1986dynamics,grizzle1985structure,yang1995structural,sundareshan1991qualitative,
  yang1996decentralized}, exchangeable or anonymous
  subsystems~\cite{madjidian2014distributed,arabneydi2015team,arabneydi2016team,elliott2013discrete}, {{}sparse connections or structure reduction \cite{arditti2020graphical,benner2004solving}, hierarchical decompositions \cite{chang1989hierarchical}}
  and patterned
systems~\cite{hamilton2012patterned}), approximate optimality (e.g., mean-field games~\cite{HCM03,HMC06,lasry2006jeux1,
  LiZhangFeng:2008}, control based on approximate aggregations
\cite{aoki1968control},  approximate distributed control
\cite{2020SuboptimalityLQ,borrelli2008distributed}, and graphon control~\cite{ShuangPeterTAC18,ShuangPeterCDC19W2}). 
%

In this paper, we propose a decomposition method for large-scale network-coupled subsystems which
relies on the spectral decomposition of the dynamic and cost coupling between
the subsystems. 
Several related approaches have been considered in the literature. An
earlier approach similar in spirit to ours is~\cite{aoki1968control}, which
considered the problem of approximating a high-dimensional system with a
low-dimensional system using what was called state aggregation. Both exact
and approximate solutions were proposed. Spectral
decomposition of large-scale systems with symmetric interconnected subsystems
have been considered in~\cite{sundareshan1991qualitative,
yang1996decentralized}. Algebraic decomposition of mean-field coupled
subsystems has been considered in~\cite{elliott2013discrete,
arabneydi2015team, arabneydi2016team}.
A key feature which distinguishes our approach from {{}these} previous work is that our approach is applicable to
models where the coupling between agents are not {homogeneous} and
that we establish optimal solutions rather than an approximate solution.  
 Another line of related work is graphical games \cite{arditti2020graphical}, \cite[Chapter 6]{lewis2013cooperative} where the coupling of the utility function depends on an underlying graph. In contrast to these, we consider a control problem and propose a different type of decomposition.
 %

\subsection{Contributions of this paper}
In this paper, we investigate a control system with multiple subsystems
connected over an undirected graph. Each subsystem has a local state and takes
a local control action. The evolution of the state of each subsystem depends
on its local state and local control as well as a weighted combination (which
we call the network field) of the states and controls of its neighbors. {{}Moreover, the weights in the network field, which are represented by a coupling matrix, may correspond to the adjacency matrix, Laplacian matrix or any other symmetric matrix that characterizes the underlying graph.}
Each
subsystem is also coupled to its multi-hop neighbors via a quadratic
cost. {{}The cost couplings can be any polynomial function of the coupling  matrix of the underlying graph.} The objective is to choose the control inputs of each subsystem to
minimize the total cost over time. 
The above model is a linear quadratic regulation problem and a centralized solution
can be obtained by solving $nd_x \times n d_x$\dash{}dimensional Riccati
equation, where $n$ is the number of subsystems and $d_x$ is the dimension of
the state of each subsystem. In this paper, we propose an alternative solution
that has low complexity and may be implemented in {{} a local manner with aggregated (or projected) state information and local information. For some particular cases, the control can be implemented in a distributed manner that relies on neighbourhood information and local information}.

The main contributions of our paper are the following: 
\begin{itemize}
\item A spectral decomposition technique is devoloped to decomposes the {{}linear quadratic control problem for network-coupled dynamical subsystems} into $L
+ 1$ decoupled {{}subproblems}, where $L$ is the rank of the coupling matrix. 

  \item {{}These $L+1$ decoupled subproblems can be solved by  solving only $L_\textup{dist}+1$ decoupled Riccati equations of dimension $d_x \times d_x$, where $L_\textup{dist}$ is the number of distinct non-zero eigenvalues of the coupling matrix and $d_x$ is the state dimension of each subsystem. In contrast, a direct centralized solution requires solving an $n d_x \times n d_x$\dash{}dimensional Riccati equation where $n$ is the number of subsystems. We note that for any coupling matrix, the inequalities $L_\textup{dist}\leq L\leq n$ always hold.} 
    Thus the method proposed in this paper leads to considerable simplifications in synthesizing optimal control laws.

  \item To implement the optimal control input, each subsystem needs to know
    the $({L+1})d_x$\dash{}dimensional vector of local components of eigen and auxiliary states (which are defined later in the paper).
    In contrast, to implement the centralized solution, each
    subsystem needs to know the $nd_x$ dimensional global state. Thus, in applications such as \cite{udell2019big,zhou2013learning,ShuangPeterCDC19W1,macarthur2008symmetry} where $L \ll n$, the
method proposed in this paper leads to considerable simplification in implementing the optimal control law. 
\item {{}The solution method is extended to solve stochastic linear quadratic control problems for network-coupled subsystems. 
\item The solution method is applied to study consensus problems to establish optimal distributed control solutions for some particular cases. }
\end{itemize}
\subsection{Notations and definitions}

We use $\BN$ and $\BR$ to denote respectively the sets of natural and real numbers. The notation $A=[a_{ij}]$ means that $a_{ij}$ is the $(i,j)$th element of the matrix $A$. For a vector $v$, $v_i$ denotes its $i$th element.  For a
matrix $A$, $A^\TRANS$ denotes its transpose. Given vectors $v^1$, \dots,
$v^n$, $\COLS(v^1, \dots, v^n)$ denotes the matrix formed by horizontally
stacking the vectors. For any ${n \in \BN}$, $\mathds{1}_{n}$ denotes the $n$-dimensional vector of ones, $\mathds{1}_{n\times n}$ denotes the $n\times n$-dimensional matrix of ones, and $I_{n}$ denotes the $n\times n$-dimensional identity matrix.  

 A pair $(A, B)$ is \emph{stabilizable} if {there exists a matrix} $L$ such that $A + BL$ is Hurwitz (i.e., all its eigenvalues have negative real parts).
A pair $(A, C)$ is \emph{detectable}  if {there exists a matrix} $F$ such that $ A+ FC$ is Hurwitz.

\section{System model and problem formulation}
\subsection{System model}\label{subsec:sys-model}

Consider a network consisting of $n$ nodes connected over an undirected
weighted graph $\mathcal{G}(\mathcal{N},\mathcal{E}, W)$, where
$\mathcal{N}=\{1,\dots, n\}$ is the set of nodes, $ \mathcal{E} \subseteq
\mathcal{N}\times \mathcal{N}$ is the unordered set of edges, and $W=[w_{ij}]\in \BR^{n\times n}$ is the weighted adjacency matrix. 
 {
Let $M=[m_{ij}]
\in \BR^{n\times n}$ be some general symmetric coupling matrix corresponding to the underlying graph $\mathcal{G}(\mathcal{N},\mathcal{E}, W)$. For instance, $M$ may represent the underlying adjacency matrix (i.e., $M=W$) or represent the underlying  Laplacian matrix (i.e., $M=\textup{diag}(W\mathds{1}_n)-W$). }
For any node $i \in \mathcal N$,
$\CN_i \DEFINED \{j\in \CN : (i,j) \in \mathcal{E}\}$ denotes the set of
neighbors of node~$i$. Note that the edge set $\mathcal{E}$ is allowed to include self-loops. Therefore the set $\CN_i$ may contain node $i$.  

The system operates in continuous time for either a finite interval~$[0,T]$ or an infinite interval $[0,\infty)$. 
A state $x_i(t) \in \BR^{d_x}$ and a control input $u_i(t) \in
\BR^{d_u}$ are associated with each node $i \in \CN$. At time $t=0$, the system
starts from an initial state $(x_i(0))_{i \in \CN}$ and for $t > 0$, the state
of node~$i$ evolves according to
\begin{equation} \label{eq:dynamics}
  \dot x_i(t) = A x_i(t) + B u_i(t) + D  x_i^\mathcal{G}(t) + E  u_i^\mathcal{G}(t),
\end{equation}
where $A$, $B$, $D$ and $E$ are matrices of appropriate dimensions and 
\begin{equation}
  x_i^\mathcal{G}(t) = \sum_{j \in \CN_i} m_{ij} x_j(t)
  \quad\text{and}\quad
  u_i^\mathcal{G}(t) = \sum_{j \in \CN_i} m_{ij} u_j(t)
\end{equation}
are the locally perceived \emph{network field} of states and control actions
at node~$i$. {{}It is assumed that all the  different subsystems have the same parameter matrices $A$, $B$, $D$ and $E$.}

We follow an atypical representation of the ``vectorized'' dynamics.
Define 
\begin{align*}
  x(t) &= \COLS(x_1(t), \dots, x_n(t)),
  \\
  u(t) &= \COLS(u_1(t), \dots, u_n(t)),
\end{align*}
as the global state and control actions of the system, and 
\begin{align*}
  x^\mathcal{G}(t) &= \COLS(x^\mathcal{G}_1(t), \dots, x^\mathcal{G}_n(t)),\\
  u^\mathcal{G}(t) &= \COLS(u^\mathcal{G}_1(t), \dots, u^\mathcal{G}_n(t)),
\end{align*}
as the global network field of states and actions. Note that $x(t), x^\mathcal{G}(t) \in
\BR^{d_x \times n}$ and $u(t), u^\mathcal{G}(t) \in \BR^{d_u \times n}$ are matrices and not vectors. 
The system dynamics may be written as
\begin{equation} \label{eq:vec-dynamics}
  \dot x(t) = A x(t) + B u(t) + D x^\mathcal{G}(t) + E u^\mathcal{G}(t).
\end{equation}
Furthermore, we may write 
\[ {
  x^\mathcal{G}(t) = x(t) M^\TRANS = x(t) M
  ~\text{and}~
  u^\mathcal{G}(t) = u(t) M^\TRANS = u(t) M.}
\]

\subsection{System performance and control objective}

At any time $t \in [0, T)$, the
system incurs an instantaneous cost 
\begin{equation}\label{eq:inst-cost}
  c(x(t), u(t)) = 
  \sum_{i \in \CN} \sum_{j \in \CN}\bigl[
    g_{ij} x_i(t)^\TRANS Q x_j(t) + 
    h_{ij} u_i(t)^\TRANS R u_j(t) 
  \bigr],
\end{equation}
and at the terminal time~$T$, the system incurs a terminal cost
\begin{equation}\label{eq:term-cost}
  c_T(x(T)) = 
  \sum_{i \in \CN} \sum_{j \in \CN}
  g_{ij} x_i(T)^\TRANS Q_T x_j(T) ,
\end{equation}
where $Q$, $Q_T$, and $R$ are matrices of appropriate dimensions and $g_{ij}$
and $h_{ij}$ are real-valued weights. 

We are interested in the following optimization problems.
\begin{problem} \label{prob:main}
  Choose a control trajectory $u \colon [0, T) \to \BR^{d_u \times n}$ to
  minimize
  \begin{equation}\label{eq:cost}
  J(u) = \int_0^T c(x(t),u(t)) dt + c_T(x(T))
  \end{equation}
  {subject to the dynamics in \eqref{eq:vec-dynamics}.}
\end{problem}

\begin{problem} \label{prob:main-infinit-horizon}
  Choose a control trajectory $u \colon [0, \infty) \to \BR^{d_u \times n}$ to
  minimize
  \begin{equation}\label{eq:inf-cost}
  J(u) = \int_0^\infty c(x(t),u(t)) dt.
  \end{equation}
  subject to the dynamics in \eqref{eq:vec-dynamics}.
\end{problem}

\subsection{Assumptions on the model}
{{}In this section, we describe the assumptions imposed on the model. Let $G = [g_{ij}]$ and $H = [h_{ij}]$.}
\begin{enumerate}
{}
  \item[\textbf{(A0)}]The weight matrices $G$ and $H$ are polynomials of $M$, i.e., 
$  G = \sum_{k = 0}^{K_G} q_k M^k$ 
and 
$
  H = \sum_{k = 0}^{K_H} r_k M^k
$
where $K_G$ and $K_H$ denote the degree of the polynomials and
$\{q_k\}_{k=0}^{K_G}$ and $\{r_k\}_{k=0}^{K_H}$ are real-valued coefficients. 
\end{enumerate}
Since $M$ is real and symmetric, it has real eigenvalues. Let $L$ denote the
rank of $M$ and $\lambda^1, \dots, \lambda^L$ denote the non-zero eigenvalues.
For ease of notation, for $\ell \in \{1, \dots, L \}$, define
\[
  q^\ell = \sum_{k=0}^{K_G} q_k (\lambda^\ell)^k
  \quad\text{and}\quad
  r^\ell = \sum_{k=0}^{K_H} r_k (\lambda^\ell)^k.
\]

\begin{enumerate}
  \item[\textbf{(A1)}] The matrices $Q$ and $Q_T$ are symmetric and positive
    semi-definite and $R$ is symmetric and positive definite. 
  \item[\textbf{(A2)}] For $\ell \in \{1, \dots, L \}$, $q^\ell$ is non-negative and  $r^\ell$ is strictly positive. Moreover $q_0\geq 0$ and $r_0>0$.
\end{enumerate}

{{}Assumption (A0) ensures that $M$, $H$ and $G$ share the same eigenvectors, which allows us to decouple the dynamics and the cost based on the same spectral decomposition (see Section \ref{sec:spectral-decomposition} for more details).} 
Assumption {(A2)} ensures that for any $y\in \BR^n$, $y^\TRANS G y\geq 0$ and $y^\TRANS H y>0$. {{}Assumptions (A1) and (A2) ensure that $G\otimes Q$ and $G\otimes Q_T$ are symmetric positive
    semi-definite, and $H\otimes R$ is symmetric positive definite, which are standard sufficient conditions for finite-horizon LQR problems to have a unique optimal solution (see for instance \cite{liberzon2011calculus}).}
    
\subsection{Some remarks on the assumptions on the cost function}
{{}
Since the subsystems (or agents) are coupled in the dynamics over an underlying graph, it is reasonable to assume that the cost structure respects the same graph structure.  The polynomials allow us to consider cost coupling structures which may involve not only the immediate neighbourhood but also multiple-hop neighbourhood connections. 

We present a few examples with different coefficients 
$\{q_k\}_{k=0}^{K_G}$ and $\{r_k\}_{k=0}^{K_H}$ below:
\begin{enumerate}
	\item If $K_G=K_H=0$, $q_0=1$, $r_0=1$ and all other coefficients are zero, then $G=H=I$. In this case, the instantaneous cost reduces to 
		\[c(x(t),u(t))= \sum_{i=1}^n [x_i(t)^\TRANS Q x_i(t)+u_i(t)^\TRANS R u_i(t)].\]
	 Thus, the problem is equivalent to the social optimal control problem where the cost is the summation of the costs of all the subsystems. 
	\item If $K_G=2$, $K_H=1$, $q_0=1, q_1=-2, q_2=1$, $r_0=1$ and all other coefficients are zero, then $G=(I-M)^2$ and $H=1$. If, furthermore,  the matrix $M=\frac1n \mathds{1}_n\mathds{1}_n^\TRANS$, then the instantaneous cost reduces to 
	\[
	\begin{aligned}
		c(x(t), u(t)) &= \sum_{i=1}^n [(x_i(t)-\bar{x}(t))^\TRANS Q (x_i(t)-\bar{x}(t))\\
		& \quad \hspace{4em} +u_i(t)^\TRANS R u_i(t)],
	\end{aligned}
	\]
	where $\bar{x}(t)\triangleq \frac{1}{n} \sum_{i=1}^n x_i(t)$. 
	Thus, the problem is equivalent to the social optimal mean field control problem (similar to the cost considered in \cite{huang2012social}). 
	\item If $K_G=2$, $K_H=0$, $q_2=1$, $r_0=1$, all other coefficients are zero, and the coupling matrix is the Laplacian matrix, then 
	\[
	G = M^2, H = I,  M = \mathcal{L}\triangleq \textup{diag}(W\mathds{1}_n - W).
	\]
	 Then the instantaneous cost reduces to
	\[
	\begin{aligned}
		(x(t), u(t)) = \sum_{i=1}^n \Big[e_i(t)^\TRANS Q e_i(t) 
		 +u_i(t)^\TRANS R u_i(t)\Big],
	\end{aligned}
	\]
	where the local state error for subsystem $i$ is given by 
	\[e_i(t) \triangleq \sum_{j\in \mathcal{N}_i}w_{ij}(x_i(t)-x_j(t)).\]
	 If, furthermore, the coupling  only appears in the cost (i.e., there are no couplings in the dynamics) and $A=0$, this structure then produces the optimal control problem that can be associated with a distributed control problem (see Section V for more details).
	 \item 
	 {$K_G$ and $K_H$ can be $\infty$ as long as the limit of the corresponding polynomial series is well defined.  Such examples include the exponential function $G= \exp({M})=\sum_{k=0}^\infty \frac{1}{k!}M^k$, and the inverse function $G= (I-\gamma M)^{-1}= \sum_{k=0}^\infty \gamma^k M^k $ (when the spectral radius $\rho(M)$ of $M$ satisfies $\rho(M)<\gamma$).} \end{enumerate} 
}


\subsection {Salient features of the model}
\begin{figure}[htb]
  \centering
  \begin{subfigure}[b]{0.45\linewidth}
    \centering
    \begin{tikzpicture}[thick,scale=0.9]
      \node [agent] at (0, 0) (3) {$3$};
      \node [agent] at (1, 0) (4) {$4$}; 
      \node [agent] at (0, 1) (2) {$2$};
      \node [agent] at (1, 1) (1) {$1$};

      \path (1) edge node[above] {$2$} (2)
            (2) edge node[left]  {$2$} (3)
            (3) edge node[below] {$1$} (4)
			(4) edge node[right]  {$1$} (1);
    \end{tikzpicture}
    \caption{A graph $\mathcal G$}
  \end{subfigure}\hfill
  \begin{subfigure}[b]{0.45\linewidth}
    \centering
    \begin{tikzpicture}[thick,scale=0.9]
      \node [agent] at (0, 0) (3) {$3$};
      \node [agent] at (1, 0) (4) {$4$}; 
      \node [agent] at (0, 1) (2) {$2$};
      \node [agent] at (1, 1) (1) {$1$};

      \path (1) edge node[left, pos=0.2] {$5$} (3)
            (2) edge node[left, pos=0.8 ]  {$4$} (4)
            (1) edge[loop right] node {$5$} (1) 
            (2) edge[loop left]  node {$8$} (2)
            (3) edge[loop left]  node {$5$} (3)
            (4) edge[loop right] node {$2$} (4) 
            ;
    \end{tikzpicture}
    \caption{$2$-hop neighborhood of $\mathcal G$}
  \end{subfigure}%
  \caption{A graph and its $2$-hop neighborhood.}
  \label{fig:multi-hop}
\end{figure}
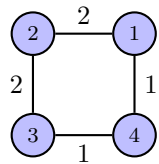
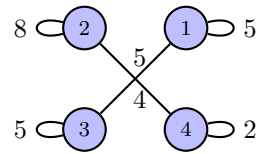

We highlight salient features of the model via an example. Consider a system
with $4$ nodes connected via a network shown in Fig.~\ref{fig:multi-hop}(a),
with 
\[
  G = q_0 I + q_1 M + q_2 M^2
  \text{ and }
  H = r_0 I + r_1 M + r_2 M^2,
\]
where $M$ and $M^2$ are the weighted adjacency matrix of the graph $\mathcal{G}$ and that of the $2$-hop
neighborhood of $\mathcal G$, respectively, given by
\[
  M = \MATRIX{ 0 & 2 & 0 & 1 \\
 2 & 0 & 2 & 0 \\
 0 & 2 & 0 & 1 \\
 1 & 0 & 1 & 0 \\ } \quad \text{and}
\quad 
  M^2 = \MATRIX{ 5 & 0 & 5 & 0 \\
 0 & 8 & 0 & 4 \\
 5 & 0 & 5 & 0 \\
 0 & 4 & 0 & 2 \\}.
\]

\subsubsection{Salient features of the dynamics}
For this example,
\begin{align*}
	&x_1^\mathcal{G}(t) = 2x_2(t) + x_4(t), \quad x_2^\mathcal{G}(t)=2x_1(t)+ 2x_3(t), \\
	&x_3^\mathcal{G}(t) = 2 x_2(t) + x_4(t), \quad x_4^\mathcal{G}(t) = x_1(t)+x_3(t).
\end{align*}
Thus, each subsystem is affected by its neighbors. The influence of each
neighbor is not homogeneous but depends on the weight associated with the
corresponding edge in the graph.  Furthermore, the network field $x^\mathcal{G}(t)$ is
not homogeneous and varies from subsystem to subsystem. 

\subsubsection{Salient features of the cost}

If $M$ is the weighted adjacency matrix of the graph $\mathcal G$, the
matrix $M^k$,
$k \in \BN$, represents the weighted adjacency matrix of the $k$-hop
neighborhood of $\mathcal G$. 
%
Thus,
\(
  G = q_0I + q_1 M + q_2 M^2
\)
means that each node has a coupling of $q_0$ with its own state, a
coupling of $q_1$ with its $1$-hop neighborhood and a coupling of $q_2$ with
its $2$-hop neighborhood. Similar interpretation holds for $H$. 
Note that
\[
  G = q_0 I +  q_1 M + q_2 M^2 = \SMATRIX{q_0+5q_2 & 2 q_1 & 5 q_2 &  q_0+q_1 \\
 2q_1 & q_0+8 q_2 & 2 q_1 & 4 q_2 \\
 5 q_2 & 2 q_1 & q_0+5 q_2 & q_1 \\
 q_1 & 4 q_2 & q_1 & q_0+2 q_2 \\
  }.
\]
Thus, the agents are not interchangeable, i.e., in general, $G_{ii} \neq
G_{jj}$ and $G_{ki} \neq G_{kj}$. 

\section{Spectral decomposition of the system}\label{sec:spectral-decomposition}

Since the weight matrix $M$ is real and symmetric, it admits a spectral
factorization. In particular, there exist non-zero eigenvalues $(\lambda^1, \dots,
\lambda^L)$ and orthonormal eigenvectors $(v^1, \dots, v^L)$ such that
\begin{equation} \label{eq:spectrum}
  M = \sum_{\ell = 1}^L \lambda^\ell v^\ell v^\ell\strut^\TRANS.
\end{equation}
In the rest of this section, we decompose the dynamics and the cost based on
the above spectral decomposition. Our decompositions may be viewed as generalizations of mean-field decompositions used in \cite{elliott2013discrete,arabneydi2015team,arabneydi2016team} to heterogenous networks.

\subsection{Spectral decomposition of the dynamics}

For $\ell \in \{1, \dots, L\}$, define eigenstates and eigencontrol actions
as 
\begin{align}
  x^\ell(t) &= x(t) v^\ell v^\ell\strut^\TRANS, 
  \label{eq:eigen-x}\\
  u^\ell(t) &= u(t) v^\ell v^\ell\strut^\TRANS,
  \label{eq:eigen-u}
\end{align}
respectively. 
Multiplying both sides of~\eqref{eq:vec-dynamics} by $v^\ell
v^\ell\strut^\TRANS$, we get
\begin{equation} \label{eq:eigen-dynamics}
  \dot x^\ell(t) = (A + \lambda^\ell D) x^\ell(t) 
  + (B + \lambda^\ell E) u^\ell(t), 
\end{equation}
where we have used the fact that $M v^\ell v^\ell\strut^\TRANS = \lambda^\ell
v^\ell v^\ell\strut^\TRANS$. Let $x^\ell_i(t)$ and $u^\ell_i(t)$ denote the
$i$-th column of these matrices, i.e., 
\begin{align*}
  x^\ell(t) = \COLS(x^\ell_1(t), \dots, x^\ell_n(t)),  \\
  u^\ell(t) = \COLS(u^\ell_1(t), \dots, u^\ell_n(t)).
\end{align*}
Therefore, the dynamics~\eqref{eq:eigen-dynamics} can be written as a collection
of decoupled ``local'' dynamics: for $i \in \CN$,
\begin{equation} \label{eq:eigen-dynamics-local}
  \dot x^\ell_i(t) = (A + \lambda^\ell D) x^\ell_i(t) 
  + (B + \lambda^\ell E) u^\ell_i(t).
\end{equation}
Using the spectral factorization~\eqref{eq:spectrum}, we may write:
\begin{align}
  x^{\mathcal{G}}(t) &= x(t) M = \sum_{\ell = 1}^L \lambda^\ell x^\ell(t), 
  \label{eq:eigen-z}\\
  u^{\mathcal{G}}(t) &= u(t) M = \sum_{\ell = 1}^L \lambda^\ell u^\ell(t).
  \label{eq:eigen-v}
\end{align}

Now, define auxiliary state and control actions as 
\[
  \breve x(t) = x(t) - \sum_{\ell = 1}^L x^\ell(t) 
  \quad\text{and}\quad
  \breve u(t) = u(t) - \sum_{\ell = 1}^L u^\ell(t).
\]
Then, by subtracting~\eqref{eq:eigen-dynamics} from~\eqref{eq:vec-dynamics}
and substituting~\eqref{eq:eigen-z} and~\eqref{eq:eigen-v}, we get
\begin{equation} \label{eq:vec-breve-dynamics}
  \dot {\breve x}(t) = A \breve x(t) + B \breve u(t).
\end{equation}

Note that $\breve x(t) \in \BR^{d_x \times n}$ and $\breve u(t) \in \BR^{d_u
\times n}$. Let $\breve x_i(t)$ and $\breve u_i(t)$ denote the $i$-th column
of these matrices, i.e., 
\begin{align*}
  \breve x(t) &= \COLS(\breve x_1(t), \dots, \breve x_n(t)),  \\
  \breve u(t) &= \COLS(\breve u_1(t), \dots, \breve u_n(t)).
\end{align*}
Therefore, the dynamics~\eqref{eq:vec-breve-dynamics} of the auxiliary state
can be written as a collection of decoupled ``local'' dynamics: 
\begin{equation} \label{eq:breve-dynamics}
  \dot {\breve x}_i(t) = A \breve x_i(t) + B \breve u_i(t), \quad i
\in \CN.
\end{equation}

The above decomposition may be summarized as follows. 
\begin{proposition} \label{prop:state-decomposition}
  The local state and control at each node $i \in \CN$ may be decomposed as 
  \begin{align}
    x_i(t) &= \breve x_i(t) + \sum_{\ell = 1}^L x^\ell_i(t),
    \label{eq:x-i} \\
    u_i(t) &= \breve u_i(t) + \sum_{\ell = 1}^L u^\ell_i(t),
    \label{eq:u-i}
  \end{align}
  where the dynamics of $\breve x_i(t)$ depend on only $\breve u_i(t)$ and are
  given by~\eqref{eq:breve-dynamics} and the dynamics of $x^\ell_i(t)$ depends
  on only $u^\ell_i(t)$ and are given by~\eqref{eq:eigen-dynamics-local}.
\end{proposition}

\subsection{Spectral decomposition of the cost}

For any $n\times n$ matrix $P = [p_{ij}]$, any $d \times n$ matrices $x =
\COLS(x_1, \dots, x_n)$, and $y  = \COLS(y_1, \dots, y_n)$, 
we use the following short hand notation:
\begin{equation}\label{def:short-hand-cost}
    \langle x , y \rangle_{P} =
    \sum_{i\in \CN} \sum_{j \in \CN} p_{ij} x_i^\TRANS y_j.
\end{equation}

\begin{proposition}\label{prop:cost-decomposition}
  The instantaneous cost may be written as
  \[
    c(x(t), u(t)) = \langle x(t), Q x(t) \rangle_G + 
    \langle u(t), R u(t) \rangle_H, 
  \]
  which can be simplified as follows:
  \begin{align*}
    \langle x(t), Q& x(t) \rangle_G \\
    &= 
     \sum_{i \in \CN} \Bigl[ q_0 \breve x_i(t)^\TRANS Q \breve x_i(t) +
     \sum_{\ell = 1}^L q^\ell x^\ell_i(t) ^\TRANS Q x^\ell_i(t)\Bigr],
    \\
    \langle u(t), R& u(t) \rangle_H \\
    &= 
     \sum_{i \in \CN} \Bigl[ r_0 \breve u_i(t)^\TRANS R \breve u_i(t) 
     +
     \sum_{\ell = 1}^L r^\ell u^\ell_i(t) ^\TRANS R u^\ell_i(t) \Bigr].
  \end{align*}
\end{proposition}
	See Appendix for the proof.

\section{The main results: Structure and synthesis of optimal control strategies}

\subsection{Finite horizon setup}
The main result for the finite horizon setup is the following.

\begin{theorem} \label{thm:main-result-control-networks}
  For $\ell \in \{1,\dots, L\}$, let $P^\ell \colon [0, T] \to \BR^{d_x \times
  d_x}$ be the solution to the backward Riccati differential equation
  \begin{equation} \label{eq:P-ell}
  \begin{aligned}
        - &\dot P^\ell(t) =
    (A + \lambda^\ell D)^\TRANS P^\ell(t) + P^\ell(t) (A + \lambda^\ell D)
    \\
     &- P^\ell(t) (B + \lambda^\ell E) (r^\ell R)^{-1} (B + \lambda^\ell E)^\TRANS P^\ell(t) +
      q^\ell Q
  \end{aligned}
  \end{equation}
  with the final condition $P^\ell(T) = q^\ell Q_T$. Similarly, let $\breve P
  \colon [0,T] \to \BR^{d_x \times d_x}$ be the solution to the backward
  Riccati differential equation
  \begin{equation} \label{eq:P-breve}
    - \dot {\breve P}(t) = A^\TRANS \breve P(t) + \breve P(t) A
    - \breve P(t) B (r_0 R)^{-1} B^\TRANS \breve P(t) + q_0 Q
  \end{equation}
  with the final condition $\breve P(T) = q_0 Q_T$. 

  Then, under assumptions \textup{{}(A0)}, \textup{(A1)} and \textup{(A2)}, the optimal control strategy for
  Problem~\ref{prob:main} is given by 
  \begin{equation} \label{eq:opt-control}
    u_i(t) = - \breve K(t) \breve x_i(t) - \sum_{\ell=1}^L K^\ell(t) x^\ell_i(t),  \quad i \in \mathcal{N},
  \end{equation}
  where 
  \begin{equation*}
    \begin{aligned}
      \breve K(t) &= (r_0R)^{-1}{B^\TRANS}\breve P(t), \\
       K^\ell(t)   &= (r^\ell R)^{-1}(B+ \lambda^\ell E)^\TRANS P^\ell(t).
    \end{aligned}
  \end{equation*}
\end{theorem}

\begin{proof}
  Consider the following collections of dynamical systems:
  \begin{itemize}
    \item Eigensystem $(\ell, i)$, 
      $\ell \in \{1, \dots, L\}$, $i \in \CN$, with state $x^\ell_i(t)$,
      control inputs $u^\ell_i(t)$, dynamics
      \begin{equation*}
        \dot x^\ell_i(t) = (A + \lambda^\ell D) x^\ell_i(t) + 
        (B + \lambda^\ell E) u^\ell_i(t),
      \end{equation*}
      and cost
      \begin{multline*}
        J^\ell_i(u^\ell_i) = 
        \int_{0}^T \bigl[
          q^\ell x^\ell_i(t)^\TRANS Q x^\ell_i(t) + 
          r^\ell u^\ell_i(t)^\TRANS R u^\ell_i(t)
        \bigr] dt
        \\ + q^\ell x^\ell_i(T)^\TRANS Q x^\ell_i(T).
      \end{multline*}

    \item Auxiliary system $i$, $i \in \CN$, with state $\breve x_i(t)$,
      control inputs $\breve u_i(t)$, dynamics
      \begin{equation*}
        \dot {\breve x}_i(t) = A \breve x_i(t) + B \breve u_i(t),
      \end{equation*}
      and cost
      \begin{multline*}
        \breve J_i(\breve u_i) = 
        \int_{0}^T \bigl[
          q_0 \breve x_i(t)^\TRANS Q \breve x_i(t) + 
          r_0 \breve u_i(t)^\TRANS R \breve u_i(t)
        \bigr] dt
        \\ + q_0 \breve x_i(T)^\TRANS Q \breve x_i(T).
      \end{multline*}
  \end{itemize}
  Note that all systems have decoupled dynamics and decoupled nonnegative cost. By
  Proposition~\ref{prop:cost-decomposition}, we have
  \[
    J(u) = \sum_{i \in \CN} \Bigl[ \breve J_i(\breve u_i) 
    + \sum_{\ell=1}^L J^\ell_i(u^\ell_i) \Bigr].
  \]
  Thus, instead of solving:
  \begin{description}
    \item[{}\textbf{(CP1)}] choose control trajectory $u \colon [0, T) \to
      \BR^{d_u \times n}$ to minimize $J(u)$,
  \end{description}
  we can equivalently solve the following optimal control problems:
  \begin{description}
    \item[{}\textbf{(CP2)}] choose control trajectory $u^\ell_i \colon [0, T) \to
      \BR^{d_u}$ to minimize $J^\ell_i(u^\ell_i)$ for $i \in \CN$, $\ell \in
      \{1, \dots, L\}$, 
    \item[{} \textbf{(CP3)}] choose control trajectory $\breve u_i \colon [0, T)
      \to \BR^{d_u}$ to minimize $\breve J_i(\breve u_i)$ for $i \in \CN$.
  \end{description}
  Given the solutions of Problems {{}(CP2)} and {{}(CP3)}, we can use 
  Proposition~\ref{prop:state-decomposition} and choose $u_i(t)$ according
  to~\eqref{eq:u-i}.

  Problems {{}(CP2)} and {{}(CP3)} are standard optimal control problems and their
  solution are given as follows. Let $P^\ell \colon [0,T] \to \BR^{d_x \times
  d_x}$ and $\breve P \colon [0, T] \to \BR^{d_x \times d_x}$ be as given
  by~\eqref{eq:P-ell} and~\eqref{eq:P-breve}. Then, for all $i \in \CN$, the
  optimal solution of {{}(CP2)} is given by
  $
    u^\ell_i(t) = K^\ell(t)
    x^\ell_i(t),  \ell \in \{1, \dots, L\},
  $
  and the solution of {{}(CP3)} is given by
  $
    \breve u_i(t) = \breve K(t) \breve x_i(t).
  $
  The result follows by combining the above two equations using~\eqref{eq:u-i}.
\end{proof}

\begin{remark}
  Based on the definition of $\breve x_i(t)$, the control in \eqref{eq:opt-control} can be equivalently written as 
  \[
    u_i(t) = - \breve K(t) x_i(t) - \sum_{\ell=1}^L \big(K^\ell(t)-\breve K(t)\big) x^\ell_i(t),
  \]
  where the first part represents a local state feedback and the second part represents offset terms proportional to eigen states. 
\end{remark}

\begin{remark}\label{rem:distinct}
  Although the eigenstates $\{x^\ell_i(t)\}_{\ell = 1}^L$ depend on the
  eigenvectors $(v^1, \dots, v^L)$, the corresponding Riccati
  equations~\eqref{eq:P-ell} only depend on the eigenvalues
  $(\lambda^1, \dots, \lambda^L)$. So, if the coupling matrix has repeated
  eigenvalues, as is the case when there are certain symmetries in the graph
  $\mathcal G$, eigendirections with the same eigenvalue have the same Riccati
  equation. Therefore, we only need to solve $L_{\textup{dist}} + 1$, Riccati
  equations, where $L_{\textup{dist}}$ denotes the number of \emph{distinct}
  non-zero eigenvalues of the coupling matrix.
\end{remark}

\begin{remark}
  The Riccati equtions~\eqref{eq:P-ell}--\eqref{eq:P-breve} are significantly
  simpler to solve than the naive centralized Riccati equation. Each Riccati
  equation in~\eqref{eq:P-ell}--\eqref{eq:P-breve} is of dimension $d_x \times
  d_x$, while the centralized Riccati equation is of dimension $n d_x \times n
  d_x$. So, even {}{if} the coupling matrix is full rank (i.e., $L = n$) {{}and all eigenvalues are distinct}, solving the
  $n$ ``one-dimensional'' Riccati
  equations~\eqref{eq:P-ell}--\eqref{eq:P-breve} is significantly simpler than
  solving one centralized ``$n$-dimensional'' Riccati equation. For graphs
  where $L \ll n$, these savings become even more drastic.
\end{remark}

\subsection{Infinite horizon setup}
{Let $Q^\frac12$ denote the symmetric positive semi-definite matrix that satisfies ${Q^\frac12}^\TRANS Q^\frac12=Q$.}
For infinite horizon problems, we further impose the following standard assumptions. 
\begin{enumerate}
  \item[\textbf{(A3)}] 
  $(A,B)$ is stabilizable and  $(q_0^\frac12 Q^\frac12,A)$ is detectable. 
  \item[\textbf{(A4)}]
  For all $\ell \in\{1,...,L\}$, $(A+ \lambda^\ell D, B+ \lambda^\ell E)$ is stabilizable and $(A+ \lambda^\ell D,{q^\ell}^\frac12 Q^\frac{1}{2})$ is detectable .
  \end{enumerate}
\begin{theorem} \label{thm:inf-horizonmain-result}
 {Suppose assumptions \textup{{}(A0)}\textendash\textup{(A4)} hold.} For $\ell \in \{1,\dots, L\}$, let $P^\ell \in \BR^{d_x \times
  d_x}$ be the unique {symmetric} positive semi-definite solution to the algebraic Riccati equation
  \begin{equation} \label{eq:inf-P-ell}
  \begin{aligned}
       0& =
    (A + \lambda^\ell D)^\TRANS P^\ell+ P^\ell (A + \lambda^\ell D)
    \\
     &\quad- P^\ell (B + \lambda^\ell E) (r^\ell R)^{-1} (B + \lambda^\ell E)^\TRANS P^\ell +
      q^\ell Q.
  \end{aligned}
  \end{equation}
  Similarly, let $\breve P \in \BR^{d_x \times d_x}$ be the unique {symmetric} positive semi-definite solution to the algebraic
  Riccati equation
  \begin{equation} \label{eq:inf-P-breve}
    0 = A^\TRANS \breve P + \breve P A
    - \breve P B (r_0 R)^{-1} B^\TRANS \breve P+ q_0 Q.
  \end{equation}
  Then the optimal control strategy for
  Problem~\ref{prob:main-infinit-horizon} is given by 
  \begin{equation} \label{eq:inf-opt-control}
    u_i(t) = - \breve K \breve x_i(t) - \sum_{\ell=1}^L K^\ell x^\ell_i(t), 
  \end{equation}
  with $\breve K = (r_0R)^{-1}{B^\TRANS}\breve P\text{ and } K^\ell   = (r^\ell R)^{-1}(B+ \lambda^\ell E)^\TRANS P^\ell.$
\end{theorem}
The proof follows along the similar lines as the proof of Theorem~\ref{thm:main-result-control-networks}. {{}Under the extra assumptions (A3) and (A4)}, one only
 needs to  replace the finite horizon costs with the infinite horizon costs and then solve  decoupled LQR problems  by solving the corresponding algebraic Riccati equations.  (A3) and (A4) ensure the existence of solutions to the algebraic Riccati equations \eqref{eq:inf-P-breve} and \eqref{eq:inf-P-ell} (see e.g. \cite{kucera1972contribution}).

\subsection{Remarks on {{}the information structure} and the implementation of the optimal strategy}

Since we are interested in regulating a deterministic system, we may implement the optimal control law either using open-loop (i.e. pre-computed) control inputs or using closed-loop (i.e. state feedback) control inputs. For both implementations, the eigenvalues $\{\lambda^\ell\}_{\ell =1}^L$ need to be known at all subsystems. 

For the open-loop implementation, one can write 
\begin{equation} \label{equ:implementation}
  u_i(t) = -\breve K(t) \breve \Phi(t,0) \breve x_i(0) 
   - \sum_{\ell=1}^L K^\ell(t) \Phi^{\ell}(t,0) x_i^\ell(0),
\end{equation}
where the state transition matrices $\breve \Phi(t,0)$ and $\Phi^\ell(t,0)$ are given by 
\begin{align}
	\breve \Phi(t,0) &= \exp\Big(\textstyle \int_0^t \big(A -
    B\breve K(\tau)\big) d\tau \Big), \\
	\Phi^\ell(t,0) &= \exp\Big(\textstyle \int_0^t \big(A+ \lambda^\ell D 
     - (B+\lambda^\ell E)K^\ell(\tau) \big)d\tau\Big). \label{eq:phi-ell}
\end{align}
Thus, to implement the control action, subsystem~$i$ needs to know $\breve x_i(0)$ and $\{x_i^\ell(0)\}_{\ell=1}^L$, which can be obtained using one of the following three information structures:

\begin{enumerate}
  \item All subsystems know the initial condition $x(0)$ and the
    eigendirections $\{v^\ell\}_{\ell=1}^L$. Using these, subsystem $i$ can
    compute $\{x_i^\ell(0)\}_{\ell=1}^L$ and $\breve x_i(0)$, and 
    implement~\eqref{equ:implementation}.

  \item Subsystem $i, i \in \CN$, knows its local initial state $x_i(0)$ and
    its local initial eigensystem states $\{x_i^\ell(0)\}_{\ell=1}^L$. Then subsystem $i$ can compute $\breve x_i(0)$ and
    implement~\eqref{equ:implementation}.

  \item All subsystems knows the initial state $\{x(0)v^\ell\}_{\ell=1}^L$. In
    addition, subsystem $i$ knows $v_i:=(v^1_i, \cdots, v_i^L)$ and its local
    initial state $x_i(0)$. Then subsystem $i$ can compute
    $\{x_i^\ell(0)\}_{\ell=1}^L$ and $\breve x_i(0)$, and 
    implement~\eqref{equ:implementation}.
\end{enumerate}

The closed-loop implementation, which is given by~\eqref{eq:opt-control} or \eqref{eq:inf-opt-control}, can be obtained by using one of the three information structures described above
with $x(0)$, $x_i(0)$ and $x_i^\ell(0)$  replaced by $x(t)$, $x_i(t)$ and
$x_i^\ell(t)$, respectively.

Furthermore, for the information structures in 2) and 3), a mixed implementation which combines open-loop and close-loop implementations can also be obtained via only replacing $x_i(0)$ by $x_i(t)$ in 2) and 3). {{}In the mixed implementation, for any subsystem $i \in \mathcal{N}$, the close-loop part corresponds to the individual state $x_i(t)$ and the open-loop part corresponds to the terms  $\{x_i^\ell(0)\}_{\ell=1}^L$ or $\{x(0)v^\ell\}_{\ell=1}^L$ which involve the aggregate of initial states of all subsystems.}

\section{Applications to consensus}

Consensus refers to a distributed coordination problem in which nodes
connected over a graph update their local states based on the states of their
neighbors. The simplest objective is for all nodes to converge to a
``consensus'' value starting from any initial state $x(0)$, i.e., 
\[
  \lim_{t\to \infty }\|x_i(t)- x_j(t) \| =0,  
  \quad \forall i,j \in \mathcal{N}.
\]
There are various consensus protocols (i.e., rules to update the state at each node as a function of the state of the nearest neighbors and its own state), which have different rates of convergence. We refer readers to \cite{saber2003consensus,olfati2007consensus,jadbabaie2003coordination,mesbahi2010graph} for an overview.
Often these consensus protocols are hand crafted based on intuitions. In this section, we show that the standard consensus protocol naturally emerges as the optimal solution of an appropriately chosen networked control problem. 

In particular, consider a (non-negatively) weighted connected  undirected graph $\mathcal{G}(\mathcal{V}, \mathcal{E}, W)$ where $W$ represents its adjacency matrix. Now consider the system dynamics
\begin{equation}\label{eq:consensus-dynamicsB=I}
	\dot{x}_i(t) = u_i(t), \quad i\in \mathcal{N}
\end{equation}
which is a special case of \eqref{eq:dynamics} with $A=0, B=I, D=0$, and $E=0$. Furthermore, consider the cost function
  \begin{equation}\label{eq:consensus-costB=I}
    c(x(t), u(t)) =\langle x(t), Q x(t)\rangle_{{M}^2}+ \langle u(t), R u(t)\rangle_{I}
  \end{equation}
where $M=\DIAG(W\mathds{1}_n) - W$ is the graph Laplacian matrix, and $Q$ and $R$ are arbitrary {symmetric} positive definite matrices. 
It is well known that the rank of the Laplacian matrix of a (non-negatively weighted) connected graph is $n-1$ and all non-zero eigenvalues are positive. Thus, $L=n-1$ for this setup. 
\begin{lemma}
 The solution to Problem \ref{prob:main-infinit-horizon} with the dynamics in \eqref{eq:consensus-dynamicsB=I} and the cost in \eqref{eq:consensus-costB=I} is given by 
\begin{align}
	u_i(t)&= - R^{-1} \Pi\sum_{\ell=1}^{n-1} \lambda^\ell x^\ell_i(t) \label{eq:consenus-eigen-decom2-B=I}, \quad i\in \mathcal{N},
\end{align}
 where $\Pi$ denotes the {symmetric} positive semi-definite solution to 
$\Pi R^{-1}\Pi = Q.$
\end{lemma}
\begin{proof}
 	Since $B=I$, $Q>0$, $R>0$, $q^\ell = (\lambda^\ell)^2>0$, $r^\ell = 1$ $q_0=0$, and $r_0=1$, (A0)\textendash(A4) are obviously satisfied.  
 	An application of Theorem \ref{thm:inf-horizonmain-result} yields the following optimal control law
\begin{equation}\label{eq:consenus-eigen-decomB=I}
	u_i(t)= -\sum_{\ell=1}^{n-1}  R^{-1}P^\ell x^\ell_i(t), \quad i\in \mathcal{N},
\end{equation}
where $P^\ell$ is the {symmetric} positive semi-definite solution to the algebraic Riccati equation
\begin{equation}\label{eq:consensus-RiccatiPl}
	0 = - P^\ell  R^{-1}  P^\ell + (\lambda^\ell)^2Q.
\end{equation}
 Note that $q_0=0$ in this example  implies the solution to the auxiliary Riccati equation in \eqref{eq:inf-P-breve} is $\breve P=0$.  Hence $\breve K =0$  in \eqref{eq:inf-opt-control} and the control law \eqref{eq:consenus-eigen-decomB=I} does not contain the auxiliary part.
Let
$\Pi = {(\lambda^\ell)}^{-1}P^\ell $. Substituting $P^\ell$  in \eqref{eq:consensus-RiccatiPl}, $\Pi$ is then given by the {symmetric} positive semi-definite solution to 
$\Pi R^{-1}\Pi = Q.$
 Hence the optimal control law is given by \eqref{eq:consenus-eigen-decom2-B=I}. ~~%
\end{proof}

Now, recall that 
\begin{align*}
	\sum_{\ell=1}^{n-1} \lambda^\ell x_i^\ell(t) &= x_i^\mathcal{G} (t) = \sum_{j\in \mathcal{N}}m_{ji}x_j= \sum_{j\in \mathcal{N}}m_{ij}x_j\\
	&= \sum_{j\in \mathcal{N}}w_{ij}(x_i-x_j)
\end{align*}
Therefore, the optimal control may be written as
\begin{equation}\label{eq:consenus-protocolB=I}
	u_i(t)= -R^{-1}\Pi \sum_{j\in \mathcal{N}}w_{ij}(x_i(t)-x_j(t)),   \quad i\in \mathcal{N}.
\end{equation}
Thus the optimal control law is the same as the standard consensus protocol in
\cite{saber2003consensus,olfati2007consensus,jadbabaie2003coordination}. A
similar result was established in~\cite[Theorem 4.6]{cao2009optimal}, using a
much more sophisticated proof argument.

\section{Generalizations to stochastic systems} \label{sec:stoc}
\subsection{Stochastic networked control problem}
In this section we consider a model similar to  Section~\ref{subsec:sys-model} but with stochastic dynamics. As before, there are $n$ subsystems that are connected over an undirected weighted graph $\mathcal{G}(\mathcal{N},\mathcal{E}, W)$, with an associated symmetric coupling matrix $M$. For any $i\in \mathcal{N}$, the state $x_i(t)$, control $u_i(t)$, and the network fields $x_i^\mathcal{G}(t)$ and $u_i^\mathcal{G}(t)$ are defined as before. The difference is that rather than being deterministic, the system dynamics are stochastic and are given by
\begin{equation}\label{eq:stochastic-dynamics}
  dx_i(t) = \Big[A x_i(t) + B u_i(t) + Dx^\mathcal{G}_i(t) + E u_i^\mathcal{G}(t)\Big]dt + F d w_i(t),
\end{equation}
for all $i\in \mathcal{N}$, where the matrices $A, B, D, E$ and $F$ are as before, $F$ is a matrix of an appropriate dimension,
 the initial states $(x_i(0))_{i\in \mathcal{N}}$ are deterministic,  
 and $\{w_i(t)\in \BR^{d_w}:i\in \mathcal{N}, t\geq 0\}$ are  standard ($d_w$-dimensional) Brownian motions that are independent across nodes.

 As before, there is an instantaneous cost $c(x(t),u(t))$ for $t\in [0,T)$, and a terminal cost $c(x(T))$,  given by \eqref{eq:inst-cost} and \eqref{eq:term-cost}.

Let $\mathcal{F}(t)$ denote the $\sigma$-algebra generated by $\{w(\tau): 0\leq \tau \leq t\}$ where $w(\tau)\triangleq \COLS\big(w_1(\tau),\ldots,w_n(\tau)\big)$. 

  We are interested in the following optimization problem. 
\begin{problem} \label{prob:stochatic}
  Choose an $\mathcal{F}(t)$-adapted  control  $u \colon [0, T) \to \BR^{d_u \times n}$
  to
  minimize
  \begin{equation}\label{eq:stochastic-cost}
  J(u) = \mathbb{E}\left[\int_0^T c(x(t),u(t)) dt +  c_T(x(T))\right],
  \end{equation}
\end{problem}
subject to the system dynamics in \eqref{eq:stochastic-dynamics} and initial conditions $(x_i(0))_{i \in \mathcal{N}}$.

\subsection{Decompositions}
Recall that $ w(t) \triangleq \COLS\big(w_1(t),\ldots,w_n(t)\big) \in \BR^{d_w\times n}$. 
We introduce the following noise processes in eigen directions and the auxiliary direction: for any $i \in \mathcal{N}$ and $\ell \in \{1,...,L\}$, 
\begin{equation*}
  \begin{aligned}
    &  w^\ell_i(t) \triangleq w(t)v^\ell v^\ell_i  ~~\text{and}~~ \breve w_i(t) \triangleq w_i(t) -\sum_{\ell =1}^L w_i^\ell(t). 
  \end{aligned}
\end{equation*}
The corresponding compact matrix representations are given by
\begin{equation*}
  \begin{aligned}
    \breve w(t) &\triangleq \COLS\big(\breve  w_1(t),\ldots,\breve  w_n(t)\big),\\
   w^\ell(t)  &\triangleq \COLS\big(w_1^\ell(t),\ldots,w_n^\ell(t)\big).
  \end{aligned}
\end{equation*}
Clearly, $w^\ell(t)= w(t)v^\ell v^\ell\strut^\TRANS$ and  $\mathds{E}[w^\ell_i]=\mathds{E} [\breve w_i] =0$.

\begin{lemma}\label{lem:w-w-ell}
The following statements hold for all $t\in[0,T]$, $i, j \in\{1,\ldots,n\}$, $\ell, h \in \{1,\ldots, L\}$:
\begin{enumerate}
  \item $\breve w_i(t)$ and $w_i^\ell(t)$ are independent.
  \item $\breve w_j(t)$ and $w_i^\ell(t)$ are  independent if and only if $v_j^\ell(v_j^\ell -v_i^\ell) = 0$. 
  \item $\breve w_i(t)$ and $\breve w_j(t)$ are independent if and only if $v_i^\ell v_j^\ell = 0$. 
  \item $w_i^\ell(t)$ and $w_j^\ell(t)$ are  independent if and only if $v_i^\ell v_j^\ell = 0$.
  \item If $i\neq j$ and $\ell \neq h$, then $w_i^\ell(t)$ and $w_j^h(t)$ are independent. 
\end{enumerate}
\end{lemma}
\begin{proof}
	Since for any fixed time $t\in [0,T]$, $\breve w_i(t)$, $\breve w_j(t)$, $w_j^\ell(t)$ and $w_k^h(t)$ are Gaussian random variables with zero mean, they are independent if and only if the covariance matrix is zero.  By explicitly computing the covariance matrices, results in Lemma~\ref{lem:w-w-ell} are verified. 
\end{proof}

Since $w_i^\ell$ and $\breve w_i$ are linear combinations of independent standard Brownian motions, they themselves are Brownian motions. 
 It is easy to verify that for $s>0, t\geq0$,
 \begin{equation*}
 \begin{aligned}
 & \text{var}(w_i^\ell(t+s)-w_i^\ell(t)) = s (v^\ell_i)^2 I_{d_w},  \\
& \text{var}(\breve w_i(t+s)-\breve w_i(t))= s\Bigl(1-\sum_{\ell=1}^L (v_i^\ell)^2\Bigr) I_{d_w}.
\end{aligned} 
 \end{equation*}
Hence the intensities of  $w_i^\ell$ and $\breve w_i$ are  $|v_i^\ell|$ and $\big(1-\sum_{\ell=1}^L (v_i^\ell)^2\big)^{\frac12}$, respectively.
 
Recall the definition of $\breve x_i(t)$, $\breve u_i(t)$, $x^\ell(t)$ and $u^\ell(t)$. 
Following arguments similar to the deterministic case, we obtain the following stochastic differential equations for the decomposed dynamics
  \begin{align}
     dx^\ell_i(t)& = \Big[(A+\lambda^\ell D) x_i^\ell(t) + (B+\lambda^\ell E)u_i^\ell(t)\Big]dt + Fdw_i^\ell(t), \label{eq:stochastic-eigen-dyn}\\
         d {\breve x}_i (t) & = \big[A \breve x_i(t) + B \breve u_i(t)\big]dt +  F d\breve w_i(t), \label{eq:stochastic-aux-dyn}
   \end{align}
for all $i \in \mathcal{N}$, $\ell \in\{1,\ldots,L\}$. 
Following the proof argument of Proposition \ref{prop:cost-decomposition}, we obtain
\begin{equation}\label{eq:stochastic-cost-sum}
  J(u) = \sum_{i \in \CN} \Bigl[ \breve J_i(\breve u_i) 
    + \sum_{\ell=1}^L J^\ell_i(u^\ell_i) \Bigr],
\end{equation}
where for all $i\in \CN$ and  $\ell \in\{1,\ldots, L\}$,
      \begin{align}
         J^\ell_i(u^\ell_i)& = 
       \mathds{E}{\Big[}  \int_{0}^T \bigl(
          q^\ell x^\ell_i(t)^\TRANS Q x^\ell_i(t) + 
          r^\ell u^\ell_i(t)^\TRANS R u^\ell_i(t)
        \bigr) dt \nonumber  \\
        &\qquad \qquad + q^\ell x^\ell_i(T)^\TRANS Q_T x^\ell_i(T)\Big],\label{eq:stochastic-eigen-cost}\\
       \breve J_i(\breve u_i) &= \mathds{E}\Big[        \int_{0}^T \bigl(
          q_0 \breve x_i(t)^\TRANS Q \breve x_i(t) + 
          r_0 \breve u_i(t)^\TRANS R \breve u_i(t)
        \bigr) dt \nonumber
        \\ & \qquad \qquad + q_0 \breve x_i(T)^\TRANS Q_T \breve x_i(T)\Big ].\label{eq:stochastic-aux-cost}
      \end{align}

\subsection{Optimal control solution}

\begin{theorem} \label{thm:main-result-control-networks-stochastic}
  Under assumptions \textup{{}(A0)}, \textup{(A1)} and \textup{(A2)}, the optimal control strategy for
  Problem~\ref{prob:stochatic} is the same as the strategy in Theorem \ref{thm:main-result-control-networks}  given by \eqref{eq:opt-control}. 
  Furthermore, the optimal cost is given by 
\begin{align}\label{eq:total-opt-sd-cost}
    V(x(0))  & = \sum_{i \in \mathcal{N}}\Big[ \breve V_i(\breve x_i(0)) + \sum_{\ell=1}^L V_i^{\ell}(x_i^\ell(0))\Big],
\end{align}
where for $i\in\mathcal{N}$ and $\ell \in\{1,...,L\}$,  
\begin{align}
   \breve V_i(\breve x_i(0)) &= \breve x_i(0)^\TRANS \breve P(0)\breve x_i(0) \nonumber\\
  & \qquad + \bigl[1-\sum_{\ell=1}^L (v_i^\ell)^2\bigr]\int_0^T \Tr\big[\breve P(t) F  F^\TRANS \big]dt, \label{eq:aux-opt-sd-cost} \\
  V_i^{\ell}(x_i^\ell(0)) &=  x_i^\ell(0)\strut^\TRANS  P^\ell(0)x_i^\ell(0)\nonumber \\
  & \qquad  + (v_i^\ell)^2\int_0^T \Tr\big[ P^\ell(t)F  F^\TRANS\big]dt \label{eq:eig-opt-sd-cost}.
\end{align}

\end{theorem}

  \begin{proof}
  
  The dynamics in \eqref{eq:stochastic-dynamics} can be decomposed into \eqref{eq:stochastic-eigen-dyn} and \eqref{eq:stochastic-aux-dyn}, and the decomposition of the cost in \eqref{eq:stochastic-cost} follows \eqref{eq:stochastic-cost-sum}, \eqref{eq:stochastic-eigen-cost} and \eqref{eq:stochastic-aux-cost}. 
   Therefore, Problem \ref{prob:stochatic} can be equivalently decomposed into the linear quadratic control problems defined by \eqref{eq:stochastic-eigen-dyn} and \eqref{eq:stochastic-eigen-cost}, and the linear quadratic control problems given by \eqref{eq:stochastic-aux-dyn} and \eqref{eq:stochastic-aux-cost}, where $i\in \mathcal{N}$.  
%
     Note that the Brownian motions are not necessarily independent across all the decoupled problems as illustrated in Lemma~\ref{lem:w-w-ell}.   
 However, following the certainty equivalence principle for {}{linear quadratic Gaussian} problems (see e.g., \cite{van1981certainty}),
  we obtain the same optimal control feedback gain as the deterministic case, which does not depend on {{}the intensity of} the Brownian motion.   This, together with the non-negativity of each term in  \eqref{eq:stochastic-cost-sum} under assumptions (A1) and (A2),  implies that solving the decomposed linear quadratic control problems independently yields the optimal feedback gain for  and hence optimal solution to Problem \ref{prob:stochatic}. Therefore, the optimal feedback gains are the same as those in Theorem \ref{thm:main-result-control-networks} for the linear quadratic control problems and the optimal control is given by \eqref{eq:opt-control}. {The optimal costs for the decomposed linear quadratic control problems are given by \eqref{eq:aux-opt-sd-cost} and \eqref{eq:eig-opt-sd-cost} (see for instance  \cite{yong1999stochastic}) and hence the optimal cost for Problem~\ref{prob:stochatic} is given by \eqref{eq:total-opt-sd-cost}. }
  \end{proof}
  
Note that  the intensity of the Brownian motion does not influence the optimal feedback gain  but it effects the optimal cost under optimal control. 

\begin{remark} The result of Theorem 3 generalizes to the infinite horizon long run average cost setup and the infinite horizon discounted cost setup in a natural manner. For each of these setups, the optimal control law will be of the same form as Theorem 3 but the control gains will be time homogeneous and determined by the solution of an algebraic Riccati equation. 
\end{remark}

\section{Illustrative Examples}
\subsection{Adjacency matrix coupling} \label{subsec:adjacency-couple}

Consider a network with $n=4$ subsystems connected over a graph $\mathcal{G}$,
as shown in Fig.~\ref{fig:example}, with its adjacency matrix as the coupling matrix $M$.
%
\begin{figure}[htb]
  \centering
\begin{tikzpicture}[thick,scale=0.9]
    \node [agent] at (0, 0) (3) {$3$};
    \node [agent] at (1, 0) (4) {$4$}; 
    \node [agent] at (0, 1) (2) {$2$};
    \node [agent] at (1, 1) (1) {$1$};

    \path (1) edge node[above] {$a$} (2)
          (2) edge node[left]  {$a$} (3)
          (3) edge node[below] {$b$} (4)
          (4) edge node[right] {$b$} (1);

    \matrix [matrix of math nodes, left delimiter={[}, right delimiter={]}]  at (3.5,0.5)
         {  0 & a & 0 & b \\
         a & 0 & a & 0 \\
         0 & a & 0 & b \\
         b & 0 & b & 0 \\} ;      
  \end{tikzpicture}
  \caption{Graph $\mathcal{G}$ with $n=4$ nodes and its adjacency matrix}
  \label{fig:example}
\end{figure}
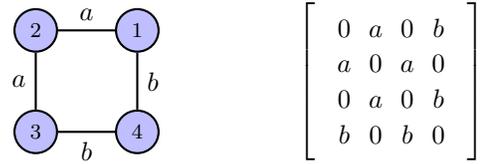
Note that $L =\text{rank}(M) =2$. Consider the following couplings in the cost
\begin{equation}\label{equ:G-H}
  G= I -2M +M^2 \quad \text{and} \quad H=I.
\end{equation}
For the ease of notation define $\rho =\sqrt{(a^2+b^2)/{2}}$ and $\theta = \tan^{-1}(b/a)$. Then it is easy to verify that the non-zero eigenvalues of $M$ are $\lambda^1 = -\rho$ and $\lambda^2 = \rho$. The corresponding eigenvectors are $v^1  =\SMATRIX{-\frac12  ~ \frac{\sin(\theta)}{\sqrt{2}} ~ -\frac12 ~\frac{\cos(\theta)}{\sqrt{2}}}^\TRANS$ and $ v^2  = \SMATRIX{\frac12 ~  \frac{\sin(\theta)}{\sqrt{2}} ~ \frac12  ~\frac{\cos(\theta)}{\sqrt{2}}}^\TRANS.$
Observe that $q^\ell=(1-\lambda^\ell)^2$ is non-negative and $r^\ell = 1$ is strictly positive, $\ell \in \{1,2\}$. Thus the model satisfies assumption~(A2). 

To illustrate how to use the result of
Theorem~\ref{thm:main-result-control-networks}, let's pick a subsystem, say
subsystem~$1$, and consider the calculations that need to be carried out at that
subsystem. 
Recall that for all $i \in \CN$, 
\(
  x_i^\ell(0) = x(0) v^\ell v^\ell_i.
\)
Thus 
\begin{align*}
  x^1_1(0) &= \frac 14 x_1(0) - \frac{\sin(\theta)}{2\sqrt{2}} x_2(0) 
  + \frac 14 x_3(0) - \frac{\cos(\theta)}{2\sqrt{2}} x_4(0),\\
  x^2_1(0) &= \frac 14 x_1(0) + \frac{\sin(\theta)}{2\sqrt{2}} x_2(0) 
  + \frac 14 x_3(0) + \frac{\cos(\theta)}{2\sqrt{2}} x_4(0). 
\end{align*}
Following the mixed implementation with information structure 3) described in Section IV-A, subsystem 1 can calculate the trajectory for $x^1_1(t)$, $x^2_1(t)$, $t\in(0,T]$ based on the initial
conditions. This together with real time local observation $x_1(t)$ yields
$\breve x_1(t)$.

Subsystem 1 solves three Riccati equations to compute $P^1(t)$, $P^2(t)$, and $\breve P(t)$ for $t\in[0,T]$, {{} where
\begin{equation} \label{eq:Riccati-Example1}
  \begin{aligned}
   - &\dot {\breve P}(t) = A^\TRANS \breve P(t) + \breve P(t) A
    - \breve P(t) B R^{-1} B^\TRANS \breve P(t) +  Q,\\
        - &\dot P^1(t) =
    (A -\rho D)^\TRANS P^1(t) + P^1(t) (A -\rho D)
    \\
     &- P^1(t) (B -\rho E)  R^{-1} (B -\rho E)^\TRANS P^1(t) +
      (1+\rho)^2 Q\\
              - &\dot P^2(t) =
    (A +\rho D)^\TRANS P^2(t) + P^2(t) (A +\rho D)
    \\ \vspace{0.3cm}
     &- P^2(t) (B +\rho E)  R^{-1} (B +\rho E)^\TRANS P^2(t) +
      (1-\rho)^2 Q\\
  \end{aligned}
  \end{equation}
with $\breve P(T) =  Q_T$, $P^1(T)= (1+\rho)^2Q_T$ and $P^2(T)=(1-\rho)^2Q_T$,}
 and then applies the optimal control action given by 
\begin{align*}
  u_1(t) = - R^{-1}\biggl(B^\TRANS \breve P(t) \breve x_1(t) &+  (B- \rho E)^\TRANS P^1(t) x_1^1(t) \\
  &+ (B+\rho E)^\TRANS P^2(t) x_1^2(t) 
   \biggr)
 \end{align*}
according to Theorem \ref{thm:main-result-control-networks}. 
 Similar implementations hold for other subsystems. 

Note that if each $x_i(t) \in \BR^{d_x}$ then $x(t) \in \BR^{4d_x}$. A naive centralized optimal solution of the above system would involve solving a $4d_x \times 4d_x$-dimensional Riccati equation. In contrast, the above solution involves solving three $d_x\times d_x$-dimensional Riccati equations. 

{{}Moreover,} these computational savings may increase with the size of the networks. For example, consider the graph $\mathcal{G}_{4c} = \mathcal{G} \otimes \mathcal{K}_c$ {{}with $4c$ nodes}, where $\mathcal{G}$ is the 4-node graph shown in Fig.~\ref{fig:example} and $\mathcal{K}_c$ is the complete graph with $c$ nodes and each edge weight is $\frac1c$ {{}where $c$ is a positive integer}. The adjacency matrix of $\mathcal{G}_{4c}$ is given by 
  $M_{4c} = M \otimes K_c,$
where $M$ and $K_c = \frac1c \mathds{1}_{c\times c}$ are the adjacency matrices of graph $\mathcal{G}$ and $\mathcal{K}_c$ respectively. 
The only non-zero eigenvalue of $K_c$ is 1. Thus, the eigenvalues of $M_{4c}$
are the same as eigenvalues of $M$. Note that the Riccati equations in Theorem
\ref{thm:main-result-control-networks} only depend on the eigenvalues. {{}So for all different graphs $\mathcal{G}_{4c}$ where $c$ can be any positive integer, the
Riccati equations are the same, given by \eqref{eq:Riccati-Example1}!}
The method proposed in Theorem \ref{thm:main-result-control-networks} would require solving the same three $d_x\times d_x$-dimensional Riccati equations in \eqref{eq:Riccati-Example1},  while a naive direct solution would require  solving a $4cd_x \times 4cd_x$-dimensional Riccati equation.

As an illustration, we consider the graph {{}$\mathcal{G}_{4c} = \mathcal{G} \otimes \mathcal{K}_c$ where $\mathcal{G}$ is given in Fig.~\ref{fig:example} with} weights $a = 2$ and $ b = 1$. Recall that $G$ and $H$ are given by \eqref{equ:G-H}.
As argued above, the matrix $M_{4c}$ has two non-zero eigenvalues and the optimal control at each subsystem can be obtained by solving {}{only} $3$ Riccati equations. Let us set $c=5$. Then $M_{20} = M \otimes \frac15 \mathds{1}_{5\times 5}$.  \\

\noindent \textbf{Example 1}: 
%
Consider Problem \ref{prob:main} with $20$ subsystems on $G_{20}$ where 
 $d_x=d_u =1$, the coupling matrix is the adjacency matrix $M_{20}$ of the graph $\mathcal{G}_{20}$, and the parameters are $A=2, B = 1, D = 3, E = 0.5$, $Q = 5, Q_T = 6, R= 2, T=2$. 
 The evolutions of eigenstates and auxiliary states along with the corresponding eigencontrols and auxiliary controls are shown in Fig.~\ref{fig:sim-example}.\\
 \begin{figure}[htb] 
\centering
  \includegraphics[width = 8.5cm,trim={2.5cm 0  2.5cm 0.8cm},clip]{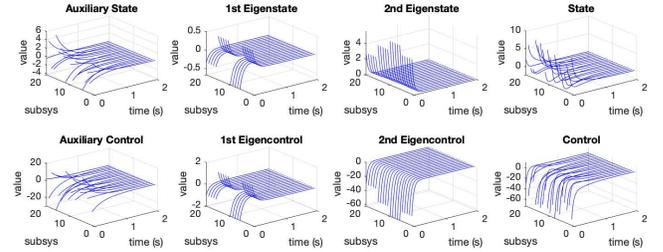}
  \caption{Numerical example under the proposed optimal control on a network of size 20 over $[0,T]$ with $T=2$.}\label{fig:sim-example}
\end{figure}

 \noindent {{}\textbf{Example 2}: 
 To illustrate the case with higher dimensional local states, we consider a network of coupled harmonic oscillators where for subsystem $i \in \mathcal{N}$, the state is given by $x_i= [\theta_i, \omega_i]^\TRANS$ representing the angle and angular velocity, and the control $u_i$ represents the input force.  
 Thus  $d_x=2$ and $d_u=1$. 
 %
Consider Problem \ref{prob:main} with $20$ coupled harmonic oscillators on $G_{20}$  where the coupling matrix is the adjacency matrix $M_{20}$ of the graph $G_{20}$ and the parameters are 
\begin{equation}\label{eq:2D-dynamics parameters}
\begin{aligned}
&A=\left[\begin{array}{cc} 0 & 10\\ -20 & 0 \end{array}\right], ~B=\left[\begin{array}{c} 0\\ 1.5 \end{array}\right],~ D=\left[\begin{array}{cc} 1 & 0\\ 0 & 1 \end{array}\right], ~R=1,
\\ 
&E=\left[\begin{array}{c} 1\\ 1 \end{array}\right], ~Q=\left[\begin{array}{cc} 6 & 0\\ 0 & 6 \end{array}\right], ~Q_T=\left[\begin{array}{cc} 5 & 0\\ 0 & 5 \end{array}\right],  ~T=2.
\end{aligned}
\end{equation}
The result is illustrated in Fig.~\ref{fig:sim-example-Adj2D}. 
\begin{figure}[htb] 
\centering
  \includegraphics[width = 8.5cm,trim={2.5cm 0  2.5cm 0.8cm},clip]{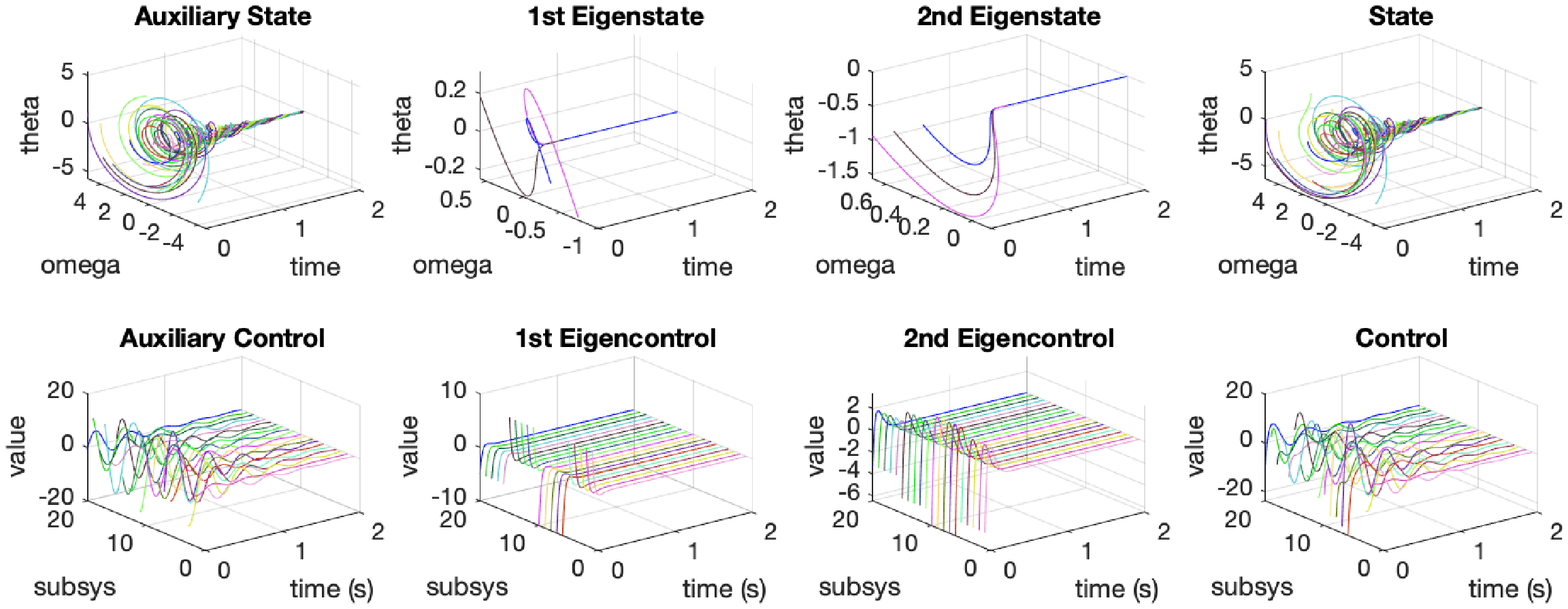}
  \caption{Numerical example under the proposed optimal control on a network of $20$ coupled harmonic oscillators  with $T=2$.}
  \label{fig:sim-example-Adj2D}
\end{figure}}
\subsection{Adjacency matrix coupling for stochastic systems}
We consider the same model as in the previous section, but assume that the
system dynamics are stochastic. In particular, we consider the graph
$\mathcal{G}_{20}$ in 
{{}Section~\ref{subsec:adjacency-couple}}.
\\

\noindent \textbf{Example 3}: 
Consider the stochastic generalization of Example 1 with
$F=1$. All other parameters are the same as those in {Example 1}. The simulation result is given in Fig.~\ref{fig:sim-example-SDE}.\\

\begin{figure}[htb] 
\centering
  \includegraphics[width = 8.5cm,trim={2.5cm 0  2.5cm 0.8cm},clip]{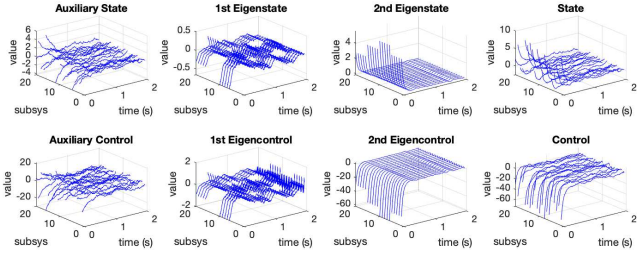}
  \caption{Numerical example with additive noise under the proposed optimal control on a network of size 20 over $[0,T]$ with $T=2$.}\label{fig:sim-example-SDE}
\end{figure}
{{}
\noindent \textbf{Example 4}: Consider the stochastic generalization of Example 2 with  
$F=\MATRIX{1&0\\0&1}$. All other parameters are the same as those in {Example 2}. The simulation result is given in Fig.~\ref{fig:sim-example-SDEAdj2D}.

\begin{figure}[htb] 
\centering
  \includegraphics[width = 8.5cm,trim={2.5cm 0  2.5cm 0.8cm},clip]{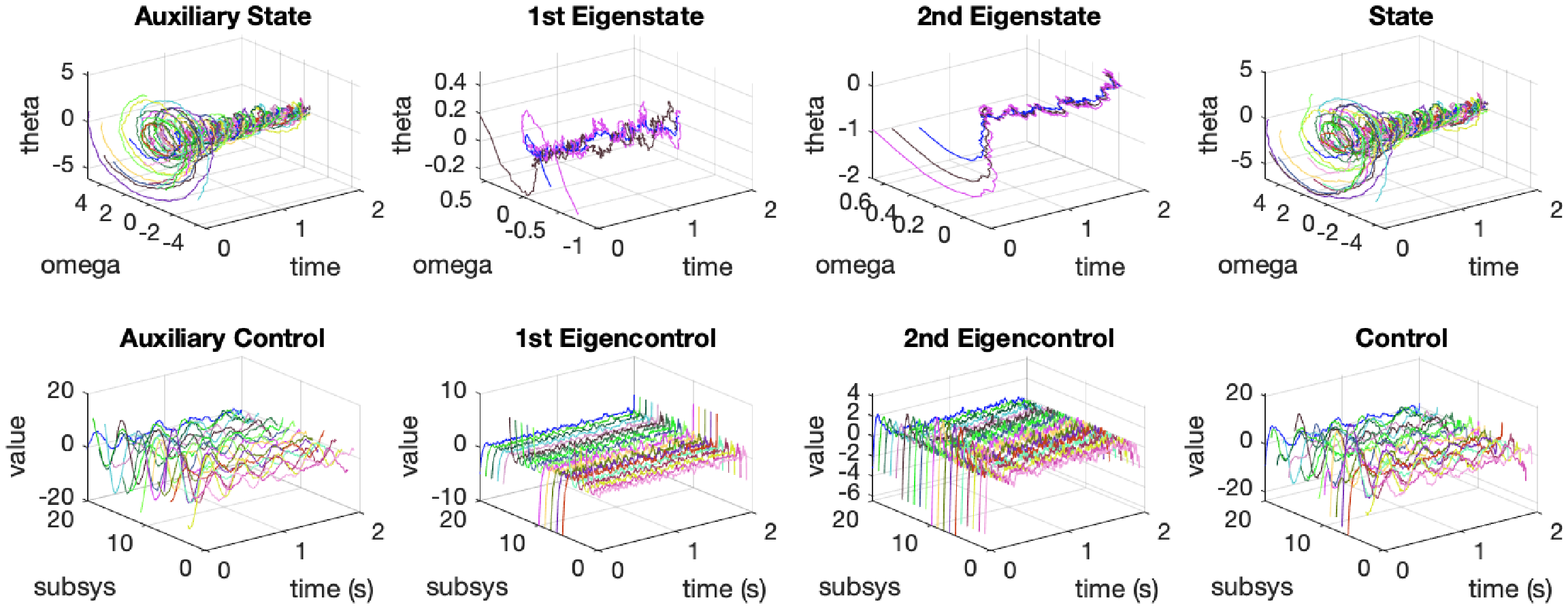}
  \caption{Numerical example with additive noise under the proposed optimal control on a network of 20 coupled harmonics oscillators with time horizon $T=2$.}\label{fig:sim-example-SDEAdj2D}
\end{figure}
}
\subsection{Laplacian matrix coupling}
We now consider examples where the coupling matrix $M$ is the Laplacian matrix of the underlying graph $\mathcal{G}_{20}$ in 
{{}Section~\ref{subsec:adjacency-couple}}.

\noindent \textbf{Example 5:} 
Consider Problem~\ref{prob:main} with $20$ subsystems on $\mathcal{G}_{20}$ where $d_x=d_u=1$, the coupling matrix is the Laplacian matrix $\mathcal{L}\triangleq\DIAG(M_{20}\mathds{1}_{20})- M_{20}$ of the graph $\mathcal{G}_{20}$,  and the parameters are
$A = 0.1$, $B = 1$, $D =E =0$, $R=0.1$, $Q=1$, $Q_T =0$, $T=2$,  $G=
\mathcal{L}^2$ and $H=I$.
The graph $\mathcal{G}_{20}$ is connected and hence the rank of $\mathcal{L}$
is 19.  However,  there are only $5$ distinct
non-zero eigenvalues. Therefore, the solution following
Theorem~\ref{thm:main-result-control-networks} requires solving $5+1$
decoupled scalar Riccati equations (see Remark~\ref{rem:distinct}): for $\ell
\in \{1, \dots, 5\}$
\begin{align}
        -\dot P^\ell(t) &=
    0.2 P^\ell(t) - 10 P^\ell(t)^2 +  (\lambda^\ell_\textup{dist})^2, ~P^\ell(T)= 0,\nonumber
    \intertext{where $\{\lambda^\ell_\textup{dist}\}_{\ell=1}^{5}$ are the
    distinct non-zero eigenvalues of $\mathcal{L}$ and } 
        -\dot {\breve P}(t) &=
    0.2 \breve P(t) - 10 \breve P (t)^2 +  0, \quad \breve P(T)=
    0\label{eq:consensus-aux-ex}.
  \end{align}
Note that, the solution to the auxiliary Riccati equation \eqref{eq:consensus-aux-ex} is $\breve P(t)=0$ for all $t \in [0,T]$, {{}which implies the control signal in the auxiliary direction should alway be zero (see the auxiliary control in Fig. \ref{fig:sim-example-SDE-Laplacian}).}  
In contrast to the above, a direct centralized solution requires
solving a $20\times 20$ dimensional matrix Riccati equation. 
The simulation
result is given in Fig.~\ref{fig:sim-example-SDE-Laplacian}.\\

 \begin{figure}[htb] 
\centering
  \includegraphics[width = 8.5cm,trim={2.5cm 0  2.5cm 0.8cm},clip]{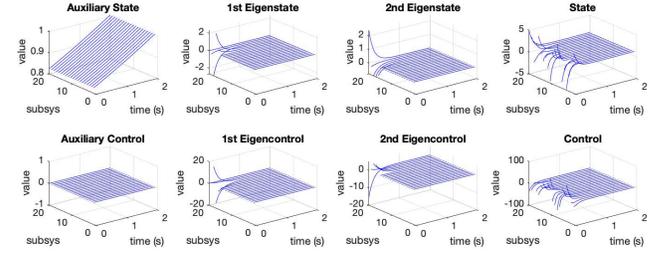}
  \caption{Numerical example with Laplacian matrix coupling under the proposed optimal control over $[0,T]$ with $T=2$.}\label{fig:sim-example-SDE-Laplacian}
\end{figure}

\noindent {{}
\textbf{Example 6:}
Consider Problem~\ref{prob:main} with 20 coupled harmonic oscillators on $\mathcal{G}_{20}$ where  $d_x=2$, $d_u=1$,  the coupling matrix is the Laplacian matrix $\mathcal{L}\triangleq\DIAG(M_{20}\mathds{1}_{20})- M_{20}$ of the graph $\mathcal{G}_{20}$,  and the parameters are given by $G= 
\mathcal{L}^2$, $H=I$ and \eqref{eq:2D-dynamics parameters}. The result is illustrated in Fig.~\ref{fig:sim-example-SDE-Lap2D}.  
\begin{figure}[htb] 
\centering
  \includegraphics[width = 8.5cm,trim={2.5cm 0  2.5cm 0.8cm},clip]{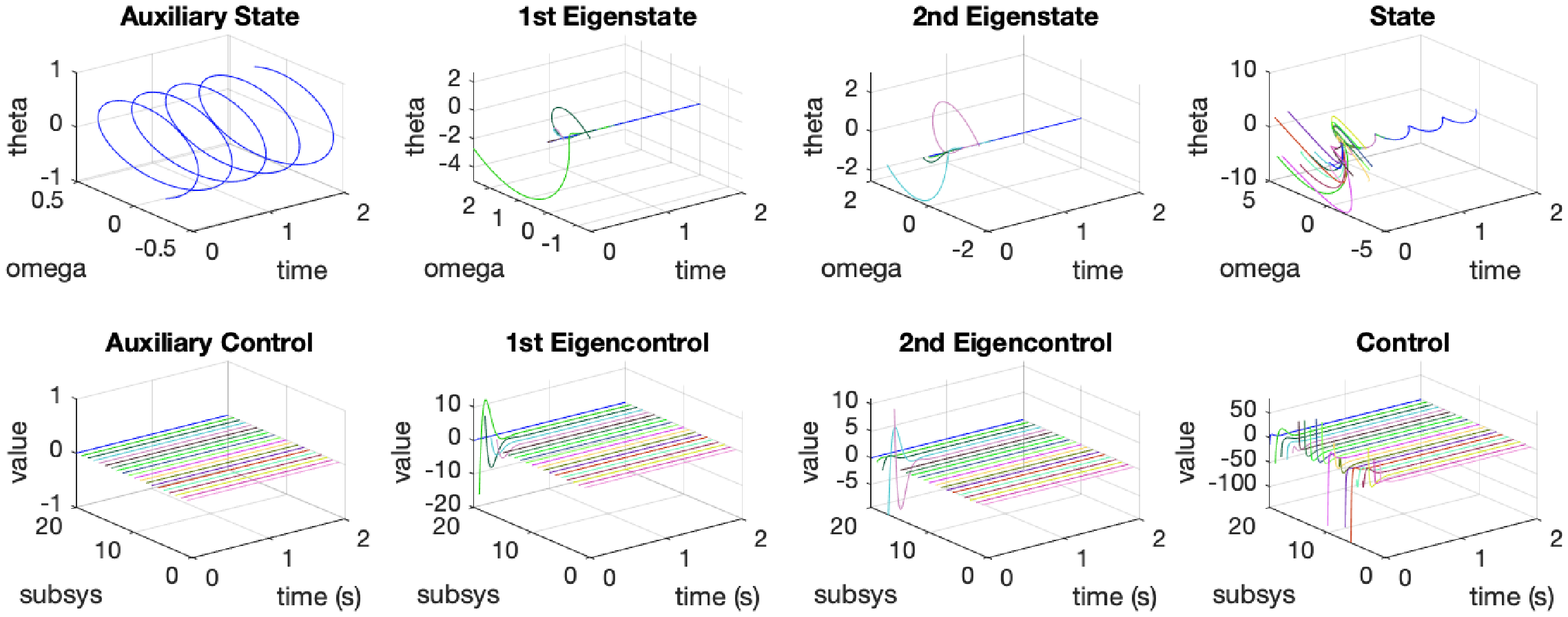}
  \caption{Numerical example with Laplacian matrix coupling under the proposed optimal control on  a network of 20 coupled harmonics oscillators over the time horizon $T=2$.}\label{fig:sim-example-SDE-Lap2D}
\end{figure}
}

\subsection{A special case: mean-field coupling} 
Suppose the graph $\mathcal{G}$ is a complete graph with all edge weights equal to $\frac{1}{n}$. Let $M$ be its adjacency matrix. Then $M = \frac1n\mathds{1}_{n\times n}$ 
has rank 1 and $\lambda_1=1$ is the only non-zero eigenvalue with the normalized eigenvector $v^1 = \frac{1}{\sqrt{n}}[1,\dots,1]^\TRANS$. Then 
$
  x^1(t)=x(t)v^1v^1\strut^\TRANS = x(t)M . 
$
Thus, the eigenstate
$  x_i^1(t) = \frac1n\sum_{j=1}^n x_j(t), i\in \CN,$
is the same for all subsystems and we denote it by $\bar x(t)$. Moreover, 
$
  q^1 = \sum_{k=0}^{K_G} q_k:= \bar{q}
$ and 
$    
  r^1 = \sum_{k=0}^{K_H} r_k:=\bar{r}.
$  
According to Theorem \ref{thm:main-result-control-networks}, the Riccati equation of eigensystem is given by 
  \begin{multline}
    -\dot{\bar P}(t) =
    (A +  D)^\TRANS \bar P(t) + \bar P(t) (A + D)
    \\
    \quad - \bar P(t) (B + E) (\bar r R)^{-1} (B + E)^\TRANS \bar P(t) +
      \bar q Q,
  \end{multline}
where $\bar P(t) := P^1(t)$ and the final condition $\bar P(T)=\bar q Q_T$.   
The Riccati equation for the auxiliary system is given by 
\begin{equation*} 
    - \dot {\breve P}(t) = A^\TRANS \breve P(t) + \breve P(t) A
    - \breve P(t) B (r_0 R)^{-1} B^\TRANS \breve P(t) + q_0 Q
  \end{equation*}
 with the final condition $\breve P(T)= q_0 Q_T$.
The optimal control strategy is given by
$u_i(t) = - \breve K(t) (x_i(t)- \bar x(t)) - \bar K(t) \bar x(t),$
where
  $\breve K(t) = (r_0 R)^{-1} B^\TRANS \breve P(t)
    \text{ and }
    \bar K(t) = (\bar r R)^{-1} (B + E)^\TRANS \bar P(t).$

The above result is similar in spirit to~\cite[Theorem 1 and Theorem
4]{arabneydi2016team}, which were derived for discrete time systems.

\section{Conclusion}

We consider the optimal control of {{}network}-coupled subsystems where the
dynamics and the cost couplings depend on an underlying undirected weighted graph. The
main idea of a low-dimensional decomposition is to project the state $x(t)$
into $L$ orthogonal eigendirections {{}where $L$ denotes the rank of the coupling matrix. This projection} generates $\{ x^\ell(t) \}_{\ell=1}^L$ and an
auxiliary state $\breve x(t) = x(t) - \sum_{\ell=1}^L x^\ell(t)$. A similar
decomposition is obtained for the control inputs. These ${L+1}$ components are
decoupled both in dynamics and cost. Therefore, the optimal control input for
each component can be obtained by solving decoupled Riccati equations. 

The proposed approach requires solving at most ${L+1}$ Riccati equations, each of
dimension $d_x \times d_x$.  
{{}If furthermore the coupling matrix has repeated non-zero eigenvalues, then the proposed approach only requires solving $L_\textup{dist}+1$ decoupled Riccati equations where $L_\textup{dist}$ (with $L_\textup{dist}\leq L \leq n$) denotes the number of distinct non-zero eigenvalues of the coupling matrix.}
In contrast, a direct centralized solution requires
solving an $n d_x \times n d_x$\dash{}dimensional Riccati equation. Thus, even
when $L_\textup{dist} = n$, the proposed approach leads to considerable computational
savings. These savings improve significantly when $L_\textup{dist} \ll n$, as is the case
for adjacency matrices for many real-world networks.

{{}Future directions of this work include: 1) similar problems where different subsystems may have different parameters in the dynamics,  2) error bounds for approximate solutions when the coupling matrix only admits an approximate low-rank representation, 3) problems with non-linear local dynamics.}
\bibliographystyle{IEEEtran}
\bibliography{IEEEabrv,mybib} 

\appendix
\section{Proof for Proposition \ref{prop:cost-decomposition}} \label{app:proof-cost-decomposition}
\subsection{Preliminary properties of the state decomposition}

\begin{lemma}
  Let $k$ be a positive integer $k$ and $\ell, \ell' \in \{1, \dots, L \}$.
  Then, we have the following:
  \begin{enumerate}
    \item[\textbf{(P1)}] $x^\ell(t) M = \lambda^\ell x^\ell(t)$ and 
      $u^\ell(t) M = 
      \lambda^\ell u^\ell(t)$.

      \vskip 0.25\baselineskip 

    \item[\textbf{(P2)}] $x^\ell(t) M^k = (\lambda^\ell)^k x^\ell(t)$ and
      $u^\ell(t) M^k = (\lambda^\ell)^k u^\ell(t)$.

      \vskip 0.25\baselineskip 

    \item[\textbf{(P3)}] $x^\ell(t) G = q^\ell x^\ell(t)$ and
      $u^\ell(t) H = r^\ell u^\ell(t)$.
      
      \vskip 0.25\baselineskip 

    \item[\textbf{(P4)}] $\breve x(t) M = 0$ and $\breve u(t) M = 0$. 

      \vskip 0.25\baselineskip 
      
    \item[\textbf{(P5)}] $\breve x(t) M^k = 0$ and $\breve u(t) M^k = 0$. 

      \vskip 0.25\baselineskip 

    \item[\textbf{(P6)}] $\breve x(t) G = q_0 \breve x(t)$ and
      $\breve u(t) H = r_0 \breve u(t)$. 

      \vskip 0.25\baselineskip 

    \item[\textbf{(P7)}] $x(t) G = q_0 \breve x(t) + \sum_{\ell=1}^L 
      q^\ell x^\ell(t)$ and 
       $u(t) G = r_0 \breve u(t) + \sum_{\ell=1}^L 
       r^\ell u^\ell(t)$.
      
      \vskip 0.25\baselineskip 

    \item[\textbf{(P8)}] $\sum_{i \in \CN} x^\ell_i(t)^\TRANS Q
      x^{\ell'}_i(t) = \delta_{\ell \ell'} \sum_{i \in \CN} x^\ell_i(t)^\TRANS
      Q x^{\ell'}_i(t)$, 
      \\[2pt] where $\delta_{\ell \ell'}$ is the Kronecker delta
      function.

      \vskip 0.25\baselineskip 

    \item[\textbf{(P9)}] 
      $\sum_{i \in \CN} x_i(t)^\TRANS Q x^{\ell}_i(t) = 
      \sum_{i \in \CN} x^\ell_i(t)^\TRANS Q x^{\ell}_i(t)$ and \\[5pt]
      $\sum_{i \in \CN} u_i(t)^\TRANS R u^{\ell}_i(t) = 
       \sum_{i \in \CN} u^\ell_i(t)^\TRANS R u^{\ell}_i(t)$
  \end{enumerate}
\end{lemma}
\begin{proof} 
  We show the result for $\breve x(t)$. The result for $\breve u(t)$ follows
  from a similar argument. 

  Since $v^1, \dots, v^L$ are orthonormal, from~\eqref{eq:spectrum} we have
  \(
    v^\ell v^\ell\strut^\TRANS M
    = \lambda^\ell v^\ell v^\ell\strut^\TRANS
  \), which implies (P1). (P2) follows immediately from (P1) and (P3) follows
  from (P2).

  (P4) follows immediately from the definition of $\breve x(t)$,
  \eqref{eq:eigen-z} and~(P1). (P5) follows immediately from (P4) and (P6)
  follows from (P5).

  (P7) follows from~\eqref{eq:x-i}, (P3) and (P6). To prove (P8), we observe
  that~\eqref{eq:eigen-x} implies that
  \begin{align}
    \sum_{i \in \CN} x^\ell_i(t)^\TRANS Q x^{\ell'}_i(t) &=
    \sum_{i \in \CN} v^\ell_i v^\ell\strut^\TRANS x(t)^\TRANS Q
    x(t) v^{\ell'} v^{\ell'}_i\strut^\TRANS
    \notag \\
    =& \Bigl( \sum_{i \in \CN} v^\ell_i v^{\ell'}_i \Bigr) 
    v^\ell\strut^\TRANS x(t)^\TRANS Q x(t) v^{\ell'}.
    \label{eq:x-x}
  \end{align}
  Since $v^1, \dots, v^L$ is orthonormal, we get
    $\sum_{i \in \CN} v^\ell_i v^{\ell'}_i 
        = v^\ell\strut^\TRANS v^{\ell'} = \delta_{\ell \ell'}.$
  Substituting this in~\eqref{eq:x-x} completes the proof of (P8). To prove
  (P9) observe that
  \begin{align}
    \sum_{i \in \CN} x_i(t)^\TRANS Q x^\ell_i(t) &= 
    \sum_{i \in \CN} x_i(t)^\TRANS Q x(t) v^\ell v^\ell_i
    \notag \\
    &= 
    \sum_{i \in \CN} v^\ell_i x_i(t)^\TRANS Q x(t) v^\ell 
    \notag \\
    &= v^\ell\strut^\TRANS x(t)^\TRANS Q x(t) v^\ell.
    \label{eq:x-x-ell}
  \end{align}
  From~\eqref{eq:x-x}, we get that the expression in~\eqref{eq:x-x-ell} is
  equal to $\sum_{i \in \CN} x^\ell_i(t)^\TRANS Q x^\ell_i(t)$.
\end{proof}

\begin{lemma} \label{lem:quadratic}
  Let $P$, $x$, and $y$ be defined in \eqref{def:short-hand-cost}. Let $P_i$
  denote the $i$-th column of $P$. Then, we can write
  \[
    \langle x, y \rangle_P = \sum_{i \in \CN} x_i^\TRANS y P_i
    \quad\text{or}\quad
    \langle x, y \rangle_P = \sum_{j \in \CN}  P_j^\TRANS x^\TRANS y_j.
  \]
\end{lemma}
\begin{proof}
  The result follows immediately from the definition of $\langle x, y
  \rangle_P$.
\end{proof}

\subsection{Proof for Proposition \ref{prop:cost-decomposition}} 
  We consider the terms depending on $x(t)$. The term depending on $u(t)$ may
  be simplified in a similar manner. 

  From~\eqref{eq:x-i} and linearity of $\langle \cdot, \cdot \rangle_G$ in
  both arguments, we get
  \begin{align}
    \langle x(t), Q x(t) \rangle_G &= 
    \Bigl\langle \breve x(t) + \sum_{\ell=1}^L x^\ell(t), 
    Q \Bigl( \breve x(t) + \sum_{\ell=1}^L x^\ell(t) \Bigr) \Bigr\rangle_G
    \notag \\
    &= \langle \breve x(t), Q \breve x(t) \rangle_G + 
    2 \Bigl\langle \sum_{\ell = 1}^L x^\ell(t), Q \breve x(t) \Bigr\rangle_G
    \notag \\
    & \quad + 
    \Bigl\langle \sum_{\ell = 1}^L x^\ell(t), Q 
    \Bigl(       \sum_{\ell = 1}^L x^\ell(t) \Bigr) \Bigr\rangle_G .
    \label{eq:x-square}
  \end{align}
  From Lemma~\ref{lem:quadratic} and (P6), the first term
  of~\eqref{eq:x-square} simplifies to
  \begin{equation} \label{eq:x-square-1}
    \langle \breve x(t), Q \breve x(t) \rangle_G
    = q_0 \sum_{i \in \CN} \breve x_i(t)^\TRANS Q \breve x_i(t),
  \end{equation}
  and the second term simplifies to
  \begin{align}
    \hskip 2em & \hskip -2em
     \Bigl\langle \sum_{\ell = 1}^L x^\ell(t), Q \breve x(t) \Bigr\rangle_G
     = q_0 \sum_{i \in \CN}  \sum_{\ell = 1}^L x^\ell_i(t)^\TRANS Q \breve x_i(t)
     \notag \\
     &= q_0 \sum_{\ell=1}^L \sum_{i \in \CN} x^\ell_i(t)^\TRANS Q 
     \Bigl( x_i(t) - \sum_{\ell' = 1}^L x^{\ell'}_i(t) \Bigr)
     \notag \\
     &\stackrel{(a)}= 
     q_0 \sum_{\ell=1}^L \sum_{i \in \CN}
     \Bigl(
       x^\ell_i(t)^\TRANS Q x^\ell_i(t) 
       -
       x^\ell_i(t)^\TRANS Q x^\ell_i(t) 
      \Bigr)
      \notag \\
      &= 0, 
  \end{align}
  where $(a)$ follows from (P8) and (P9).
  From Lemma~\ref{lem:quadratic} and (P3), the third term
  of~\eqref{eq:x-square} simplifies to
  \begin{align}
    \hskip 2em & \hskip -2em
    \Bigl\langle \sum_{\ell = 1}^L x^\ell(t), Q 
    \Bigl(       \sum_{\ell = 1}^L x^\ell(t) \Bigr) \Bigr\rangle_G
    \notag \\
    &=
    \sum_{i \in \CN} \sum_{\ell=1}^L x^\ell_i(t)^\TRANS Q 
    \Bigl( \sum_{\ell'=1}^L q^{\ell'} x^{\ell'}_i(t) \Bigr)
    \notag \\
    &= \sum_{\ell = 1}^L \sum_{i \in \CN} x^\ell_i(t)^\TRANS Q 
    \Bigl( \sum_{\ell'=1}^L q^{\ell'} x^{\ell'}_i(t) \Bigr)
    \notag \\
    &\stackrel{(b)}=
    \sum_{\ell = 1}^L \sum_{i \in \CN} 
    q^\ell x^\ell_i(t) ^\TRANS Q x^\ell_i(t),
    \label{eq:x-square-3}
  \end{align}
  where $(b)$ follows from (P8). We get the result by
  substituting~\eqref{eq:x-square-1}--\eqref{eq:x-square-3}
  in~\eqref{eq:x-square}.

\end{document}